\documentclass[twocolumn, aps, pra, showpacs, superscriptaddress, floatfix, nofootinbib, 10pt]{revtex4-2}
\usepackage{amsfonts, amsmath, amssymb, amsthm}
\makeatletter
\def\amsbb{\use@mathgroup \M@U \symAMSb}
\makeatother
\usepackage{ae,lmodern}
\usepackage[T1]{fontenc}
\usepackage[mathscr]{euscript}
\usepackage{graphicx}
\usepackage{color}
\definecolor{darkred}{RGB}{200, 0, 0}
\definecolor{darkgreen}{RGB}{0, 100, 0}
\definecolor{darkblue}{RGB}{0, 0, 200}
\usepackage[normalem]{ulem}
\usepackage[toc,page]{appendix}
\usepackage[ocgcolorlinks, colorlinks=true, linkcolor=darkblue, citecolor=darkred]{hyperref}
\usepackage{latexsym}
\usepackage{ifthen}
\usepackage{natbib}
\usepackage{verbatim}
\usepackage{psfrag}

\DeclareMathOperator{\trj}{Tr}
\DeclareMathOperator{\spec}{spec}

\newcommand{\ave}[2][]{#1\langle #2 #1\rangle}
\newcommand{\abs}[2][]{#1| #2 #1|}

\newcommand{\nbox}[2][9]{\hspace{#1pt} \mbox{#2} \hspace{#1pt}}

\newcommand{\betaave}{\beta_{\textnormal{ave}}}
\newcommand{\tran}[0]{^\textnormal{\tiny{T}}}
\newcommand{\cH}{\mathcal{H}}

\newcommand{\cL}{\mathcal{L}}
\newcommand{\cM}{\mathcal{M}}

\newcommand{\cR}{\mathcal{R}}

\newcommand{\sB}{\mathscr{B}}

\newcommand{\sF}{\mathscr{F}}

\newcommand{\sP}{\mathscr{P}}

\newcommand{\Tr}[1]{\textnormal{Tr}{\left(#1\right)}}
\newcommand{\Trr}[2]{\textnormal{Tr}_{#1}{\left(#2\right)}}
\def \diracspacing {0.7pt}
\newcommand{\ket}[1]{| \hspace{\diracspacing} #1 \rangle} 
\newcommand{\braket}[2]{\langle #1 \hspace{\diracspacing} | \hspace{\diracspacing} #2 \rangle} 
\newcommand{\ketbra}[2]{| \hspace{\diracspacing} #1 \rangle \langle #2 \hspace{\diracspacing} |} 
\newcommand{\ketbraq}[1]{\ketbra{#1}{#1}} 
\newcommand{\unit}{1\!\!1}
\newcommand{\mean}[2][]{#1\langle #2 #1\rangle}
\newcommand{\CHSH}{\textnormal{CHSH}}
\newcommand{\GHZ}{\textnormal{GHZ}}
\newcommand{\Mer}{\textnormal{Mer}}
\newcommand{\sigABdef}[1]{\sigma_{ A' B_{1} } #1:= ( \Gamma_{A} \otimes \unit_{ B_{1} } )( \tau_{ A B_{1} } )}
\newcommand{\sigBCdef}[1]{\sigma_{ B_{2} C' } #1:= ( \unit_{ B_{2} } \otimes \Gamma_{C} )( \tau_{ B_{2} C } )}
\newcommand{\gfunc}[1]
{
\begin{equation*}
g(x) := \frac{1}{2} + \frac{1}{2} \cdot  \frac{ x - x^{*} }{ 2 \sqrt{2} - x^{*} }#1
\end{equation*}
}
\newcommand{\betacrit}{2.689}
\newtheorem*{thm:self-testing-BSM-robust-2}{Theorem \ref*{thm:self-testing-BSM-robust}}
\newcommand{\BSMtheorem}[1]
{
Let the initial state shared by Alice, Bob and Charlie be of the form
\begin{equation*}
\tau_{A B_{1} B_{2} C} = \tau_{A B_{1}} \otimes \tau_{B_{2} C}
\end{equation*}
and let $\sB := ( B_{B_{1}B_{2}}^{b} )_{b = 0}^{3}$ be a four-outcome measurement acting on $\cH_{B_{1}} \otimes \cH_{B_{2}}$.
\ifthenelse{#1=0}
{If there exist measurements for Alice and Charlie such that the resulting statistics conditioned on $b$ exhibit the maximal violation of the $\CHSH_{b}$ inequality, then there exist completely positive and unital maps $\Lambda_{B_{1}} : \cL( \cH_{ B_{1} } ) \to \cL( \cH_{A'} ), \Lambda_{B_{2}} : \cL( \cH_{ B_{2} } ) \to \cL( \cH_{C'} )$ for $\abs{ A' } = \abs{ C' } = 2$ such that
\begin{equation*}
\big( \Lambda_{B_{1}} \otimes \Lambda_{B_{2}} \big) \big( B_{ B_{1} B_{2} }^{b} \big) =
\Phi_{A'C'}^{b}
\end{equation*}
for $b \in \{0, 1, 2, 3\}$.}
{Let $\tau_{B_{1}} = \trj_{A} \tau_{ A B_{1} }$, $\tau_{B_{2}} = \trj_{C} \tau_{ B_{2} C }$ be the marginal states and $p_{b} := \ave{ \tau_{B_{1}} \otimes \tau_{B_{2}}, B_{ B_{1} B_{2} }^{b} }$ be the probability of Bob observing outcome $b$. Suppose that the statistics of Alice and Charlie conditioned on that outcome give the violation of $\beta_{b}$ of the $\CHSH_{b}$ inequality and that the average violation satisfies $\betaave := \sum_{b} p_{b} \, \beta_{b} > 2$. If we define $q := g( \betaave )$ for
\gfunc{,}
where $x^{*} := ( 16 + 14 \sqrt{2} ) / 17$, then the quality of the real measurement $\sB$ as a simulation of the Bell-state measurement $\Phi$ satisfies
\begin{equation*}
Q( \sB, \Phi ) \geq \frac{1}{ 2 ( 1 + \eta^{*} ) } \min_{ v \in [0, \eta^{*} ] } \bigg[ \frac{2q - 1}{ \sqrt{ 1 - v^{2} } } + \frac{1}{ 1 + v } \bigg],
\end{equation*}
where $\etastardef$.
}
}
\newcommand{\taubdef}
{
\begin{equation*}
p_{b} \tau_{AC}^b = \trj_{B_{1} B_{2}} \big[ ( \unit_{AC} \otimes B_{ B_{1} B_{2} }^{b} ) ( \tau_{A B_{1}} \otimes \tau_{B_{2} C} ) \big].
\end{equation*}
}
\newcommand{\extractabilitylb}[1]
{
\begin{equation*}
F ( ( \Gamma_{A} \otimes \Gamma_{C} ) ( \tau_{AC}^b ), \Phi_{A'C'}^b ) \geq g( \beta_{b})#1
\end{equation*}
}
\newcommand{\exdefbi}{( \Lambda_{B_{1}} \otimes \Lambda_{B_{2}} ) ( F_{ B_{1} B_{2} }^{j} ) = P_{ B_{1}' B_{2}' }^{j}}
\newcommand{\Qunipartite}{Q( \sF, \sP ) := \frac{1}{\abs{A'}} \max_{ \Lambda } \sum_{j = 1}^{d} \ave[\big]{ \Lambda( F_{A}^{j} ), P_{A'}^{j} }}
\newcommand{\Qsep}{Q_{\textnormal{sep}}}
\newcommand{\Qsepdef}{Q_{\textnormal{sep}}(\sP) := \frac{1}{\abs{B_1'} \cdot \abs{B_2'}} \max_{ \sF \in \cM_{\textnormal{sep}} } \sum_{j = 1}^{d} \ave[\big]{ F_{B_1'B_2'}^{j}, P_{B_1'B_2'}^{j} }}
\newcommand{\sqroot}[2]{\big( #1_{A}^{-1/2} \otimes \unit \big) #1_{AB_{#2}} \big( #1_{A}^{-1/2} \otimes \unit \big)}
\newcommand{\etastardef}{\eta^{*} := 2 \sqrt{ q (1 - q ) }}

\newcommand{\be}{\begin{eqnarray}}
\newcommand{\ee}{\end{eqnarray}}
\newtheorem{prop}{Proposition}
\newtheorem{lemma}{Lemma}
\newtheorem{theorem}{Theorem}

\newtheorem{step}{Step}

\def \shortversion {0}
\begin{document}
\title{Self-testing entangled measurements in quantum networks}
\author{Marc Olivier Renou}
\affiliation{D\'epartement de Physique Appliqu\'ee, Universit\'e de Gen\`eve, CH-1211 Gen\`eve, Switzerland}
\author{J\k{e}drzej Kaniewski}
\affiliation{Center for Theoretical Physics, Polish Academy of Sciences, Al.~Lotnik{\'o}w 32/46, 02-668 Warsaw, Poland}
\affiliation{QMATH, Department of Mathematical Sciences, University of Copenhagen, Universitetsparken 5, 2100 Copenhagen, Denmark}
\author{Nicolas Brunner}
\affiliation{D\'epartement de Physique Appliqu\'ee, Universit\'e de Gen\`eve, CH-1211 Gen\`eve, Switzerland}
%
\date{\today}
%

\begin{abstract}
Self-testing refers to the possibility of characterizing an unknown quantum device based only on the observed statistics. Here we develop methods for self-testing entangled quantum measurements, a key element for quantum networks. Our approach is based on the natural assumption that separated physical sources in a network should be considered independent. This provides a natural formulation of the problem of certifying entangled measurements. Considering the setup of entanglement swapping, we derive a robust self-test for the Bell-state measurement, tolerating noise levels up to $\sim 5\%$. We also discuss generalizations to other entangled measurements.

\end{abstract}
\maketitle
\ifthenelse{\equal{\shortversion}{1}}
{}{}
\emph{Formalizing the problem.---}Previous works have developed methods for self-testing entangled states, as well as sets of local measurements. For instance, observing the maximal quantum violation of the CHSH Bell inequality implies that the local measurements are essentially a pair of anti-commuting Pauli observables~\cite{popescu92a, kaniewski17a}. Hence, what is certified in this case is how two measurements relate to each other, but not what they are individually. 

In the present work, we focus on a different problem, namely to self-test a single measurement featuring entangled eigenstates. For clarity, we first formalize the problem without considering the specific structure of the eigenstates. Let $\sP = \big( P_{A'}^{j} \big)_{j = 1}^{d}$ be the ``ideal'' $d$-outcome measurement acting on a Hilbert space $\cH_{A'}$ and $\sF = \big( F_{A}^{j} \big)_{j = 1}^{d}$ be the ``real'' measurement acting on $\cH_{A}$. Our goal is formalize the notion that $\sP$ and $\sF$ are in some sense equivalent. In the standard tomographic (device-dependent) setting we would simply require that all the measurement operators are the same (implying that $\cH_{A'} = \cH_{A}$). Clearly, this cannot work in the device-independent setting as, for instance, one cannot even certify that two Hilbert spaces have the same dimension.

In the device-independent scenario the best we can hope for - what we achieve here for the BSM - is to certify that $\sF$ is at least as powerful as $\sP$, i.e.~that $\sF$ \emph{can be used to simulate} $\sP$. We say that $\sF$ is \emph{capable of simulating} $\sP$ if there exists a completely positive unital map $\Lambda : \cL( \cH_{A} ) \to \cL( \cH_{A'}
 )$ such that
\begin{equation}
\label{eq:exact-definition-unipartite}
\Lambda( F_{A}^{j} ) = P_{A'}^{j}
\end{equation}
for all $j$ and let us justify this definition by providing an explicit simulation procedure. Note that the dual map $\Lambda^{\dagger} : \cL( \cH_{A'} ) \to \cL( \cH_{A} )$ is completely positive and trace-preserving, i.e.~it is a quantum channel. Given an unknown state $\sigma$ acting on $\cH_{A'}$, we would like to obtain the statistics produced under the ideal measurement $\sP$. It suffices to apply the channel $\Lambda^{\dagger}$ to $\sigma$ and perform the real measurement $\sF$.
Indeed, the probability of observing the outcome $j$ is given by
\begin{equation*}
\Pr[j] = \trj \big( \Lambda^{\dagger}(\sigma) F_{A}^{j} \big) 
%
= \trj \big( \sigma P_{A'}^{j} \big),
\end{equation*}
matching the statistics of the ideal measurement.
It is important that the quantum channel $\Lambda^{\dagger}$ is \emph{universal}, i.e.~it does not depend on the input state $\sigma$.

The second key aspect of our problem is the fact that the measurement eigenstates are entangled. Clearly, this is only meaningful given that there is a well-defined bipartition for the measurement device. This point is addressed in a very natural way in the context of quantum networks. Consider as in Fig. 1 a network featuring three observers (Alice, Bob, and Charlie), and two separated sources: the first source distributes a quantum system to Alice and Bob (represented by a state on $ \cH_{A}\otimes\cH_{B_1}$), while the second source distributes a system to Bob and Charlie (given by a state on $ \cH_{B_2}\otimes\cH_{C}$). It is natural to assume that, due to their separation, the two sources are independent from each other, an assumption also made in recent works discussing Bell nonlocality in networks (see e.g. \cite{branciard,fritz}). Hence, Bob receives two well-defined physical systems (one from Alice and one from Charlie), which ensures that his measurement device features a natural bipartition, specifically  $ \cH_{B} = \cH_{B_1}\otimes\cH_{B_2}$. 

The problem of self-testing an entangled measurement can now be formalized as follows. Given an ideal measurement for Bob $\sP = \big( P_{B_1'B_2'}^{j} \big)_{j = 1}^{d}$ acting on $\cH_{B_1'} \otimes \cH_{B_2'}$ and a measurement $\sF = \big( F_{B_1 B_2}^{j} \big)_{j = 1}^{d}$ acting on $\cH_{B_1} \otimes \cH_{B_2}$ we say that $\sF$ is capable of simulating $\sP$ if there exist completely positive unital maps $\Lambda_{B_1} : \cL( \cH_{B_1} ) \to \cL( \cH_{B_1'} )$ and $\Lambda_{B_2} : \cL( \cH_{B_2} ) \to \cL( \cH_{B_{2}'} )$ such that
\begin{equation}
\label{eq:exact-definition-bipartite}
\exdefbi
\end{equation}
for all $j$. Next we look at specific scenarios and show that observing certain statistics allows us to self-test an entangled  measurement, i.e.~conclude that the real measurement applied in the experiment is capable of simulating some specific ideal measurement.

\begin{figure}
\center
\includegraphics[width=0.9\columnwidth]{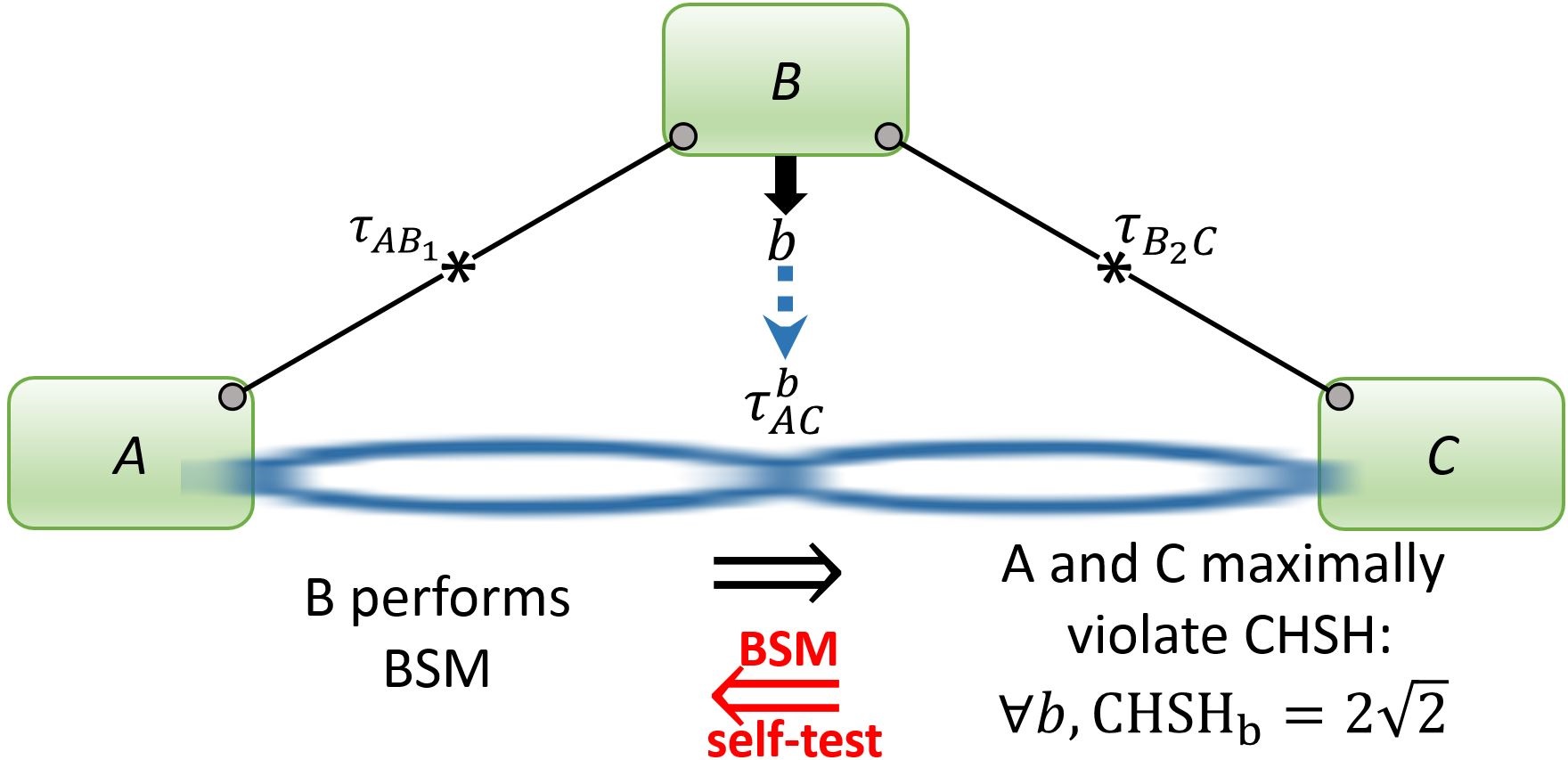}
\caption{We consider an entanglement swapping scenario for self-testing the Bell-state measurement (BSM). It is well-known that, by performing the BSM, Bob can maximally entangle the two systems of Alice and Charlie, that were initially uncorrelated. Hence Alice and Charlie can observe maximal violation of the CHSH Bell inequality. Here we prove the converse statement, namely that the observation of $\CHSH_b=2\sqrt{2}$ for all $b$ (see text) necessarily implies that Bob's measurement is equivalent to the BSM. In the second part of the paper, we show that this result can be made robust to noise.
}
\label{Protocole}
\end{figure}


\emph{Self-testing the Bell-state measurement.---}Let us start by presenting a simple procedure for self-testing the BSM. The four Bell states (maximally entangled two-qubit states) are given by
\begin{align*}
&\ket{\Phi^0} := \ket{\phi^+} = \frac{\ket{00}+\ket{11}}{\sqrt{2}},
&\ket{\Phi^1} := \ket{\phi^-} = \frac{\ket{00}-\ket{11}}{\sqrt{2}},\\
&\ket{\Phi^2} := \ket{\psi^+} = \frac{\ket{01}+\ket{10}}{\sqrt{2}},
&\ket{\Phi^3} := \ket{\psi^-} = \frac{\ket{01}-\ket{10}}{\sqrt{2}}
\end{align*}
and the BSM corresponds to $\Phi= (\Phi^{b} )_{b = 0}^{3}$ with $\Phi^b :=\ketbraq{\Phi^b}$.

Our certification procedure relies on the task of entanglement swapping \cite{ekert}, see Fig. 1. The goal is to generate entanglement between two initially uncorrelated parties (Alice and Charlie) by using an additional party (Bob) who is independently entangled with each of them. Specifically, Alice and Bob share a maximally entangled state $\ket{\phi^+}_{AB_1}\in \cH_{A}\otimes\cH_{B_1}$, and similarly Bob and Charlie share $\ket{\phi^+}_{B_2C}\in \cH_{B_2}\otimes\cH_{C}$. 
When Bob performs the BSM and obtains outcome $b$, the state of Alice and Charlie is projected to $\Phi^{b}_{ AC}$. That is, for each outcome $b$, Alice and Charlie now share one of the four Bell states. If the outcome $b$ is communicated to (say) Alice she can apply a local unitary operation on her qubit, so that she now shares with Charlie a specific Bell state.

Our self-testing procedure is based on the observation that for every outcome of Bob, the conditional state shared between Alice and Charlie can be self-tested and, moreover, we can choose their local measurements to be independent of $b$. If Alice and Charlie perform the local measurements $A_{0} := \sigma_{z}$, $A_{1} := \sigma_{x}$, $C_{0} := ( \sigma_{z} + \sigma_{x} ) / \sqrt{2}$, $C_{1} := ( \sigma_{z} - \sigma_{x} ) / \sqrt{2} $, their statistics \emph{conditioned on} $b$ will maximally violate \emph{some} CHSH inequality. More specifically, we will observe $\CHSH_{b} = 2\sqrt{2}$, where

\begin{align*}
\CHSH_{0}&:=\mean{A_0 C_0}+\mean{A_0 C_1}+\mean{A_1 C_0}-\mean{A_1 C_1},\\
\CHSH_{1}&:= \mean{A_0 C_0}+\mean{A_0 C_1}-\mean{A_1 C_0}+\mean{A_1 C_1},\\
\CHSH_2&:=-\CHSH_{1}, \quad \CHSH_3:=-\CHSH_{0}.
\end{align*}
It turns out that observing these statistics necessarily implies that Bob performs a BSM, according to the definition given in Eq.~\eqref{eq:exact-definition-bipartite}.

\begin{theorem}
\label{thm:self-testing-BSM-exact}
\BSMtheorem{0}
\end{theorem}
While a complete proof is given in Appendix~B, we only briefly sketch the argument here. From now on, it is important to distinguish the ideal system (denoted with primes) from the real system (without primes). Let $p_{b}$ be the probability of Bob observing the outcome $b$ and let $\tau_{AC}^b$ be the normalized state between Alice and Charlie conditioned on that particular outcome, i.e.
\taubdef
Since every conditional state $\tau_{AC}^b$ maximally violates some CHSH inequality, the standard self-testing result \cite{bardyn,kaniewski16b} tells us that for each $b$ there exist local extraction channels that produce a maximally entangled state of two qubits. In fact, since the extraction channels are always constructed from the local observables which do not depend on $b$, there exists a single pair of extraction channels $\Gamma_{A} : \cL( \cH_{A} ) \to \cL( \cH_{A'} ), \Gamma_{C} : \cL( \cH_{C} ) \to \cL( \cH_{C'} )$, which always produces the ``correct'' maximally entangled state, i.e.
\begin{equation}
\label{eq:perfect-extraction}
( \Gamma_{A} \otimes \Gamma_{C} ) \big( \tau_{AC}^b \big) =
\Phi_{A'C'}^{b}.
\end{equation}
Since applying these extraction channels commutes with the measurement performed by Bob, we can formally construct the state $\sigABdef{}$.
Since this is a positive semidefinite operator satisfying $\sigma_{A'} = \unit/2$, it can be rescaled to become the Choi state of a unital map from $\cL( \cH_{B_{1}} )$ to $\cL( \cH_{A'} )$.
More specifically, we choose the Choi state of $\Lambda_{B_{1}}$ to be $2 \sigma_{ A' B_{1} }\tran$, where $\tran$ denotes the transpose in the standard basis. Similarly, we define $\sigBCdef{}$ and choose the Choi state of $\Lambda_{B_{2}}$ to be $2 \sigma_{ B_{2} C' }\tran$.
The final result of the theorem follows from a straightforward computation, in which we show that the state $\sigma_{A'C'}^b$ shared between Alice and Charlie after Bob measured $B^b_{B_1B_2}$ is by definition the image of $B^b_{B_1B_2}$ by the renormalized Choi map associated to $\sigma_{A'B_1}^T\otimes\sigma_{B_2C'}^T$ (see Appendix F).

Before discussing self-testing of the BSM in the noisy case, we present two natural generalizations of Theorem~1.

\emph{Generalizations.---}The first extension is a self-testing of the ``tilted'' Bell-state measurement (tilted BSM), featuring four partially entangled two-qubit states as eigenstates
\begin{align*}
&\ket{\phi^+_\theta}=c_\theta\ket{00}+s_\theta\ket{11},
&\ket{\phi^-_\theta}=s_\theta\ket{00}-c_\theta\ket{11},\\
&\ket{\psi^+_\theta}=c_\theta\ket{01}+s_\theta\ket{10},
&\ket{\psi^-_\theta}=s_\theta\ket{01}-c_\theta\ket{10},
\end{align*}
where $0<\theta\leq\pi/4$ and $c_\theta=\cos{\theta}$, $s_\theta=\sin{\theta}$. The self-test is again based on entanglement swapping, with initially two shared maximally entangled states. The difference is that Bob's measurement now prepares partially entangled states for Alice and Charlie, which they can self-test \cite {yang,bamps} via the maximal violation of the tilted CHSH inequalities \cite{tilted}; see Appendix~C for details.  

Our second generalization is for a three-qubit entangled measurement, with eight eigenstates given by the GHZ states
\begin{align*}
&\ket{\GHZ_0^\pm}=\frac{\ket{000}\pm\ket{111}}{\sqrt{2}},
&\ket{\GHZ_1^\pm}=\frac{\ket{011}\pm\ket{100}}{\sqrt{2}},\\
&\ket{\GHZ_2^\pm}=\frac{\ket{101}\pm\ket{010}}{\sqrt{2}},
&\ket{\GHZ_3^\pm}=\frac{\ket{110}\pm\ket{001}}{\sqrt{2}}.
\end{align*}
The self-testing procedure involves a star network of 4 observers. The central node (Rob) shares a maximally entangled state with each of the three other observers. For each of the 8 measurement outcomes, Rob's measurement prepares a GHZ state shared by the three other observers, which can be self-tested \cite{mckague,kaniewski16b} via the maximal violation of the Mermin Bell inequalities \cite{mermin}; see Appendix~C for details.

\emph{Robust self-testing of the Bell-state measurement.---}So far, we have shown that the BSM can be self-tested in the noiseless case, i.e. when Alice and Charlie observe the maximal CHSH violation. However, from a practical point of view, it is of course crucial to investigate whether such a result can be made robust to noise. In this section, we derive a noise-robust version of Theorem 1. 

Recall that given the ideal measurement $\sP$ acting on $\cH_{A'}$ and the real measurement $\sF$ acting on $\cH_{A}$ we say that the real measurement $\sF$ is capable of simulating the ideal measurement $\sP$ if there exists a completely positive unital map $\Lambda : \cL( \cH_{A} ) \to \cL( \cH_{A'}
 )$ such that $\Lambda( F_{A}^{j} ) = P_{A'}^{j}$ for all $j$. Since in the device-independent setting one cannot certify non-projective measurements, let us from now on assume that $\sP$ is a projective measurement and we define the \emph{quality of $\sF$ as a simulation of $\sP$} as
\begin{equation*}
\Qunipartite,
\end{equation*}
where $\abs{A'}$ is the dimension of the ideal Hilbert space $\cH_{A'}$,  $\ave{\cdot,\cdot}$ is the Hilbert-Schmidt inner product and the maximization is taken over completely positive unital maps $\Lambda : \cL( \cH_{A} ) \to \cL( \cH_{A'}
 )$. The quantity $Q(\sF, \sP)$ is well-defined as long as $\sF$ and $\sP$ have the same number of outcomes and $Q( \sF, \sP) \in [0, 1]$ (see Appendix~D). Moreover, since $Q( \sF, \sP) = 1$ iff $\sF$ is capable of simulating $\sP$, it is justified to think of $Q$ as a measure of how good the simulation is. This definition naturally generalizes to the case where $\sF$ and $\sP$ act jointly on two subsystems as
\begin{align*}
Q( \sF, \sP ) :=& \frac{1}{\abs{B_{1}'} \cdot \abs{B_{2}'}}\cdot\\ &\max_{ \Lambda_{B_{1}}, \Lambda_{B_{2}} } \sum_{j = 1}^{d} \ave[\big]{ ( \Lambda_{B_{1}} \otimes \Lambda_{B_{2}} ) ( F_{ B_{1} B_{2} }^{j} ), P_{ B_{1}' B_{2}' }^{j} }
\end{align*}
where the maximization is taken over completely positive unital maps $\Lambda_{X} : \cL( \cH_{X} ) \to \cL( \cH_{X'} )$ for $X= B_1, B_2$. Since we are interested in certifying entangled measurements, we assume that the ideal measurement $\sP$ contains at least one entangled operator. The threshold value $Q_{\textnormal{sep}}(\sP)$, above which we can conclude that the real measurement is entangled, is simply the largest value of $Q$ achievable when the real measurement is separable, i.e.
\begin{equation*}
\Qsepdef,
\end{equation*}
where $\cM_{\textnormal{sep}}$ is the set of separable measurements acting on $\cH_{ B_{1}' } \otimes \cH_{ B_{2}' }$. Since $\sP$ contains some entangled measurement operators, we have $\Qsep(\sP) < 1$ and clearly exceeding this threshold guarantees that at least one measurement operator of $\sF$ is entangled. For the special case of rank-1 projective measurements a simple to evaluate upper bound on $\Qsep(\sP)$ can be derived in terms of the Schmidt coefficients of the measurement operators (see Appendix~D). For the BSM this bound turns out to be tight and we conclude that $\Qsep(\Phi) = \frac{1}{2}$.

Let us now state the robust version of Theorem~\ref{thm:self-testing-BSM-exact} and sketch the proof.

\begin{theorem}
\label{thm:self-testing-BSM-robust}
\BSMtheorem{1}
\end{theorem}
The final bound, plotted as a function of $\betaave$ in Fig.~\ref{fig:robust-bsm}, is non-trivial for $\betaave \gtrsim \betacrit$ (corresponding to $\sim 5 \% $ of noise) which certifies that the measurement is entangled. As the proof is rather technical, we give a brief overview below, but defer a formal argument to Appendix~E.
\begin{figure}
	\centering
	\includegraphics[scale=0.94]{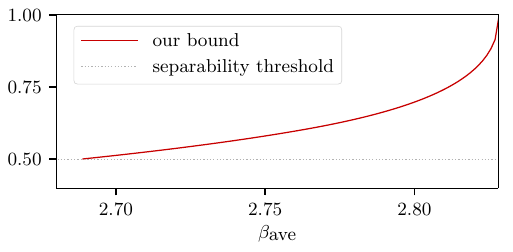}
	\caption{Lower bound on the quality of the unknown measurement proven in Theorem~\ref{thm:self-testing-BSM-robust}. The average CHSH violation of $\betaave \gtrsim \betacrit$ certifies that the measurement is entangled.}
	\label{fig:robust-bsm}
\end{figure}
The proof follows the argument given for the exact case until Eq.~\eqref{eq:perfect-extraction}, which in the noisy case must be replaced by an approximate statement. The standard construction of extraction channels \cite{kaniewski16b} yields channels $\Gamma_{A}, \Gamma_{C}$ such that the fidelity between the extracted state and the corresponding Bell state satisfies
\extractabilitylb{}
for all $b $. 
As before we define $\sigABdef{}$ and $\sigBCdef{}$, which allows us to write
\begin{align}
p_{b} F ( ( \Gamma_{A} &\otimes \Gamma_{C} ) ( \tau_{AC}^b ), \Phi_{A'C'}^b )\nonumber\\
&= p_{b} \ave{( \Gamma_{A} \otimes \Gamma_{C} ) ( \tau_{AC}^b ), \Phi_{A'C'}^b }\nonumber\\
&= \ave{ \sigma_{A' B_{1}} \otimes \sigma_{B_{2} C'}, \Phi_{A'C'}^b \otimes B_{ B_{1} B_{2} }^{b} } \geq p_{b} g( \beta_{b}).\label{eq:noisyextraction}
\end{align}
As the marginals of $\sigma_{A'}$ and $\sigma_{C'}$ are no longer guaranteed to be uniform, $\sigma_{A'B_1},\sigma_{B_2C'}$  cannot be rescaled to become Choi states of unital channels. A more complicated construction yields $\lambda_{A'B_{1}}\tran$ and $\lambda_{B_{2} C'}\tran$ with uniform marginals on subsystems $A'$ and $C'$, but their closeness to $\sigma_{A'B_{1}}$ and $\sigma_{B_{2} C'}$ depends on the bias of the marginals $\sigma_{A'}$ and $\sigma_{C'}$. Fortunately, this bias can be estimated from the observed Bell violation. Applying the unital channels corresponding to $\lambda_{A'B_{1}}$ and $\lambda_{B_{2} C'}$ yields
\begin{align*}
\ave{ ( \Lambda_{ B_{1} } \otimes \Lambda_{ B_{2} } ) ( B_{ B_{1} B_{2} }^{b}& ) , \Phi_{A'C'}^b }\\ &= \ave[\big]{ \lambda_{A' B_{1}}\tran \otimes \lambda_{B_{2} C'}\tran , \Phi_{A'C'}^b \otimes B_{ B_{1} B_{2} }^{b} },
\end{align*}
which we can relate to the observed Bell violation through Eq.~\eqref{eq:noisyextraction}. In Appendix E.2 we explain how the same approach can be used to derive robust results for the GHZ measurement discussed before.

\ifthenelse{\equal{\shortversion}{1}}
{}{}

\section*{Appendix A: The formal swap isometry}

In this appendix, we introduce the formal swap gate $S_{X,Z}$ and swap channel $\Gamma_{X,Z}$ defined for two operators $X, Z$ of a Hilbert space $\cH$. Note that these gates perform a swap only under some conditions (e.g. anti-commutation on the support of an input state) given in Lemma~\ref{LemmaSwap} stated below. In the following, $\cH'$ is a qubit Hilbert space, $S_{X,Z}$ maps any state $\ket{\psi'}\otimes\ket{\psi}\in\cH'\otimes\cH$ into $\cH'\otimes\cH$ (we write again $S_{X,Z}$ the corresponding maps over density matrices) and $\Gamma_{X,Z}$ maps any operator $\rho\in\mathcal{L}(\cH$) into $\mathcal{L}(\cH')$. 
This two transformations are introduced in Figure~\ref{Swap} and read
\begin{align}
S_{X,Z}(\ket{0}\ket{\psi})&=\frac{1}{2}\left[\ket{0}(\unit+Z)\ket{\psi}+\ket{1} X(\unit-Z)\ket{\psi} \right],\label{defS0}\\
S_{X,Z}(\ket{1}\ket{\psi})&=\frac{1}{2}\left[\ket{0}X(\unit+Z)\ket{\psi}+\ket{1}(\unit-Z)\ket{\psi} \right],\label{defS1}\\
\Gamma_{X,Z}(\rho)&=\Trr{\cH}{S_{X,Z}(\ketbra{0}{0}\otimes\rho)}.
\end{align}
This formal swap idea was already introduced in \cite{ST2} to self-test states. Here, we introduce a slightly different operator which simplifies the formulation of Lemma~\ref{LemmaSwap}.

Let $X', Z'$ be the usual Pauli matrices over the qubit space $\cH'$. We have:
\begin{lemma}\label{LemmaSwap}
Assume that $X^2=Z^2=\unit$ and $X,Z$ anti-commute over the support of a state $\ket{\psi}\in\cH$. 
Acting with $X$ (resp. $Z$) before applying $S_{X,Z}$ is equivalent to an action of $X'$ (resp.  $Z'$) after applying $S_{X,Z}$, i.e.
\begin{align}
S_{X,Z}\cdot X\cdot\ket{{\psi'}}\otimes\ket{\psi}={X'}& \cdot S_{X,Z}\cdot\ket{{\psi'}}\otimes\ket{\psi}\label{SwapEq1},\\
S_{X,Z}\cdot Z\cdot\ket{{\psi'}}\otimes\ket{\psi}={Z'}& \cdot S_{X,Z}\cdot\ket{{\psi'}}\otimes\ket{\psi}\label{SwapEq2}.
\end{align}
The conjugation of $X, Z$ with $X$ (resp. $Z$) in the definition of $S_{X,Z}$, which maps $X$ into $-X$ (resp. $Z$ into $-Z$) is equivalent to an action of $X'\otimes X$ (resp. $Z'\otimes Z$) after applying $S_{X,Z}$, i.e.
\begin{align}
S_{-X,Z}\cdot \ket{{\psi'}}\otimes\ket{\psi}=Z'\otimes Z&\cdot S_{X,Z}\cdot \ket{{\psi'}}\otimes\ket{\psi}\label{SwapEq3},\\
S_{X,-Z}\cdot \ket{{\psi'}}\otimes\ket{\psi}=X'\otimes X&\cdot S_{X,Z}\cdot \ket{{\psi'}}\otimes\ket{\psi}\label{SwapEq4}.
\end{align}
\end{lemma}

\begin{figure}
\center
\includegraphics[scale=0.3]{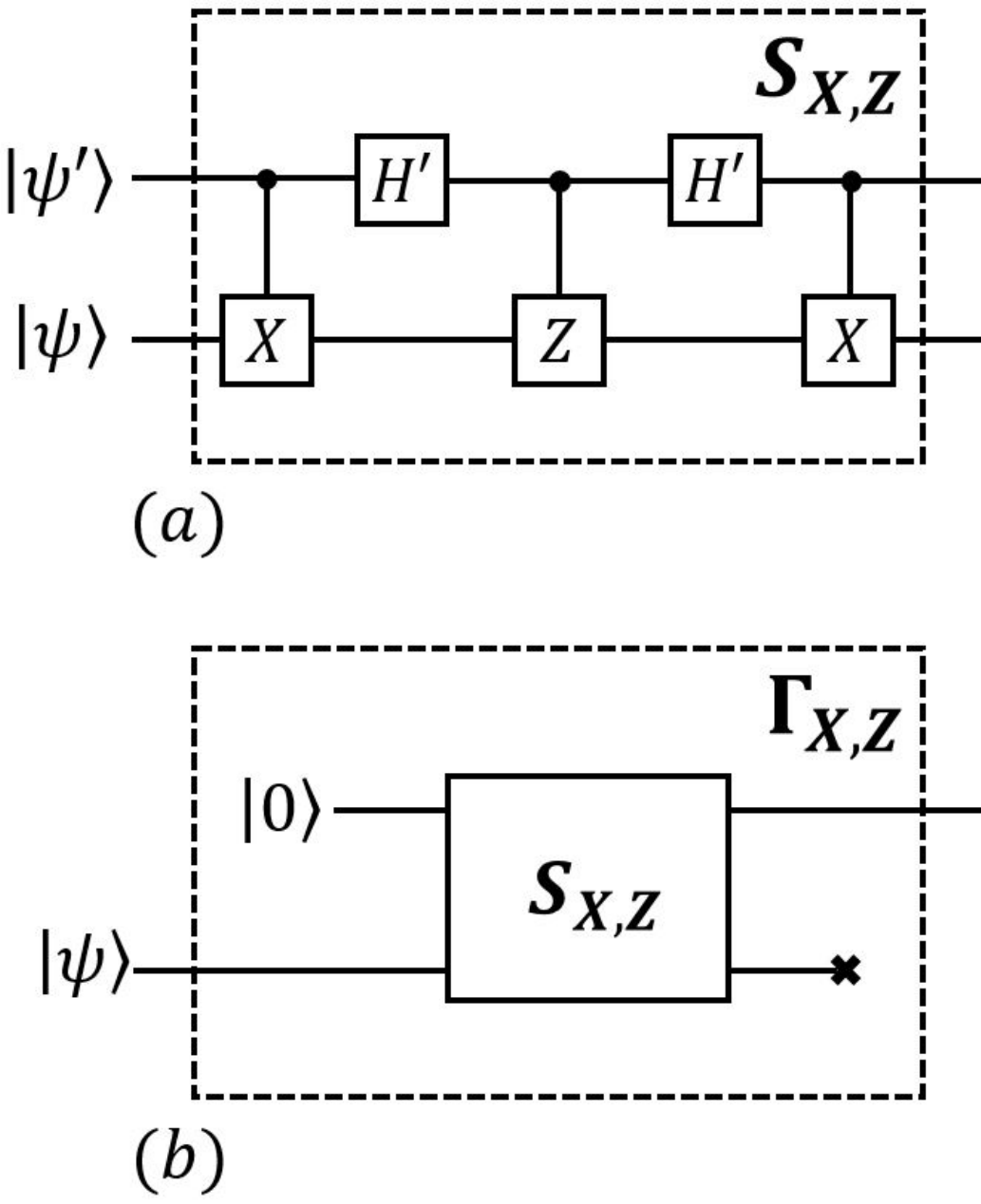}[b]
\caption{(a) Swap gate $S_{X,Z}$ constructed out of two operators $X$, $Z$ which anti-commute over the support of a state $\ket\psi\in\cH$. $S_{X,Z}$ has two entries, a qubit $\ket{\psi'}\in\cH'$ and a state $\ket{\psi}\in\cH$. $H'$ is the Hadamard gate. (b)  Swap isometry $\Gamma_{X,Z}$. It corresponds to the swap gate in which the qubit $\ket{\psi'}$ is initialized to $\ket{0}$ and the output state in $\cH$ is traced out.}\label{Swap}
\end{figure}

\begin{proof}
By linearity, we can restrict ourselves to $\ket{{\psi'}}\in\{\ket{0},\ket{1}\}$.
Then, this can directly obtained from Eq.~(\ref{defS0})~and~(\ref{defS1})~by adding $X$ or $Z$ in front of $\ket{{\psi'}}\otimes\ket{\psi}$ or substituting $X$ into $-X$ or $Z$ into $-Z$ and using the anti-commutation rules.
\end{proof}

Remark that Lemma~\ref{LemmaSwap} still holds when $\ket{\psi}$ is replaced by a density matrix, possibly defined over a larger Hilbert space.

\section*{Appendix B: Self-testing of the Bell-state measurement}\label{appendix_tiltedBS}
We now come to the proof of the main theorem of our letter. Let us introduce the formal Pauli matrices
\begin{align*}
Z_A&=A_0, &X_A&=A_1,\\
Z^*_C&=\frac{C_0+C_1}{\sqrt{2}}, &X_C^*&=\frac{C_0-C_1}{\sqrt{2}},\\
Z_C&=\mathrm{r}(Z^*_C)|\mathrm{r}(Z^*_C)|^{-1}, &X_C&=\mathrm{r}(X^*_C)|\mathrm{r}(X^*_C)|^{-1},
\end{align*} 
where for a Hermitian operator $O^*$, $O=\mathrm{r}(O^*)$ is the regularized operator i.e. the same operator in which all zero eigenvalues have been replaced by 1.
We have:
\begin{theorem}\label{thm:self-testing-BSM-exact}
\BSMtheorem{0}
\end{theorem}
The proof is in two steps. 
We first prove that the Hilbert spaces of both Alice and Charlie can be replaced by qubit Hilbert spaces. In that case, after Bob's measurement, Alice and Charlie share one of the four Bell states. 
Then we choose the Choi state of $\Lambda_{B_{1}}$ and $\Lambda_{B_{2}}$ to be proportional to the shared state between Alice/Bob and Bob/Charlie, and show that they satisfy the desired properties in order to extract the BSM.

\begin{step}\label{step1CHSH}
Let $\Gamma^0_A=\Gamma_{X_A,Z_A}$ and $\Gamma^0_C=\Gamma_{X_C,Z_C}$.
We introduce the reduced states
\begin{align*}
\sigma_{A'B_1}:&=\Gamma^0_{A}(\tau_{AB_1}), ~~~~~~~~
\sigma_{B_2C'}:=\Gamma^0_{C}(\tau_{B_2C}),\\
\sigma^b_{A'C'}:&=\left(\Gamma^0_{A}\otimes \Gamma^0_{C}\right)(\tau^b_{AC})\\
&=4\Tr{B^b_{B_1B_2} \left( \sigma_{A'B_1}\otimes\sigma_{B_2C'}\right)    },
\end{align*}
where $\tau^b_{AC}$ is the state shared between Alice and Charlie after Bob measured $b\in\{0, \cdots, 3\}$.
Then, we have
\begin{equation}\label{apeendixequationstep2}
\sigma^b_{A'C'}=\Phi^b_{A'C'}.
\end{equation}
\end{step}
\begin{proof}
We first prove the case $b=0$, which corresponds to the self-test of a maximally entangled state of two qubits \cite{ST2}. We briefly sketch it here for completeness.
After Bob measured $b=0$, the state $\tau^0_{AC}$ maximally violates the $\CHSH_0$ inequality.
Hence $Z_A, X_A$ and $ Z_C, X_C$ anti-commute and square to identity over the support of $\tau^0_{AC}$ (e.g. see \cite{bamps}). Hence we can apply Lemma~\ref{LemmaSwap} respectively to $(Z_A, X_A)$ and $(Z_C, X_C)$ which introduces the two qubit Hilbert spaces $\cH_{A'}$, $\cH_{C'}$ and the maps 
\begin{align*}
S^0_A&=S_{X_A,Z_A}, &S^0_C&=S_{X_C,Z_C},\\
 \Gamma^0_A&=\Gamma_{X_A,Z_A}, &\Gamma^0_C&=\Gamma_{X_C,Z_C}.
\end{align*}
We write $\cH_{AC}:=\cH_{A}\otimes\cH_{C}$ and $\cH_{A'C'}:=\cH_{A'}\otimes\cH_{C'}$.
Let 
\begin{equation}
W_0=A_0C_0+A_0C_1+A_1C_0-A_1C_1
\end{equation} 
be the Bell operator acting over $\cH_{AC}$ and 
\begin{equation}
W'_0=A'_0C'_0+A'_0C'_1+A'_1C'_0-A'_1C'_1
\end{equation} 
the ideal Bell operator acting over $\cH'_{AC}$. 
Eq.~(\ref{SwapEq1})~and~(\ref{SwapEq2}) show that 
\begin{equation*} S^0_A\otimes S^0_C(W_0\cdot\ketbra{00}{00}\otimes\tau^0_{AC})=W'_0\cdot S^0_A\otimes S^0_C(\ketbra{00}{00}\otimes\tau^0_{AC}).
\end{equation*}
 Hence $\sigma^0_{A'C'}=\left(\Gamma^0_{A}\otimes \Gamma^0_{C}\right)(\tau^0_{AC})$ maximally violates CHSH$_0$.
It is straightforward to show that the eigenvalue $2\sqrt{2}$ of the operator $W'_0$ is non degenerated, with associated eigenvector $\Phi^0_{A'C'}$. 
Hence $S^0_A\otimes S^0_C(\ketbra{00}{00}\otimes\tau^0_{AC})$ is a product state between $\cH_{AC}$ and $\cH_{A'C'}$, and $\sigma^0_{A'C'}=\Phi^0_{A'C'}$.
\\

Let us now prove Step~\ref{step1CHSH} for $b=1$, the other cases being similar. Post selecting the statistics over $b=1$, we have a maximal violation of $\CHSH_1$. Hence, as $\CHSH_1$ is linked to $\CHSH_0$ by the relabeling $A_1\rightarrow -A_1$, considering 
\begin{align*}
S^1_A&=S_{-X_A,Z_A}, &S^1_C&=S^0_C,\\
\Gamma^1_A&=\Gamma_{-X_A,Z_A}, &\Gamma^1_C&=\Gamma^0_C,
\end{align*}
 we can exploit the proof for $b=0$ (with $X_A$ replaced by $-X_A$), which gives that $S^1_A\otimes S^1_C(\ketbra{00}{00}\otimes\tau^1_{AC})$ is a product state and 
\begin{equation*}
\left(\Gamma^1_{A}\otimes \Gamma^1_{C}\right)(\tau^1_{AC})=\Phi^0_{A'C'}.
\end{equation*}
With Lemma~\ref{LemmaSwap}, we have
\begin{align*}
\sigma^1_{A'C'}&=\left(\Gamma^0_{A}\otimes \Gamma^0_{C}\right)(\tau^1_{AC})\\
&=\Trr{\cH_{AC}}{\left( S^0_{A}\otimes S^0_{C}\right)(\ketbra{00}{00}\otimes\tau^1_{AC})}\\
&=Z'_A\Trr{\cH_{AC}}{Z_A\left( S^1_{A}\otimes S^1_{C}\right)(\ketbra{00}{00}\otimes\tau^1_{AC})Z_A}Z'_A\\
&=Z'_A\Phi^0_{AC}Z'_A=\Phi^1_{A'C'}.
\end{align*}
\end{proof}
\begin{step}\label{step2_CHSH}
Let $\Lambda_{B_1}:\cL(\cH_{B_1})\rightarrow\cL(\cH_{A'})$ and $\Lambda_{B_2}:\cL(\cH_{B_2})\rightarrow\cL(\cH_{C'})$ be respectively the Choi-Jamio{\l}kowski maps associated to the operator $2\sigma_{A'B_1}$ and $2\sigma_{B_2C'}$. These maps are unital and
\begin{equation}
\Lambda_{B_1}\otimes\Lambda_{B_2} (B^b_{B_1B_2})=\Phi^b_{A'C'},
\end{equation}
which proves Theorem~\ref{thm:self-testing-BSM-exact}.
\end{step}

\begin{proof}
$\Lambda_{B_1}$ is unital iff it maps $\unit_{B_1}$ to $\unit_{A'}$. 
By definition of the Choi-Jamio{\l}kowski map, we have
\begin{align*}
\Lambda_{B_1}(\unit_{B_1})&=\Trr{B_1}{\sigma_{A'B_1}}\\
&=\Trr{B_1B_2C'}{\sigma_{A'B_1}\otimes\sigma_{B_2C'}}\\
&=\sum_b\Trr{B_1B_2C'}{B^b_{B_1B_2}\sigma_{A'B_1}\otimes\sigma_{B_2C'}}\\
&=\frac{1}{4}\sum_b \Trr{C'}{\sigma^b_{A'C'}}=\unit_{A'},
\end{align*}
where we used that $\sum_b B^b_{B_1B_2}=\unit$ and $\sigma^b_{A'C'}=\Phi^b_{A'C'}$.

Moreover, according to the definition of the Choi-Jamio{\l}kowski isomorphism (see Appendix~\ref{appendix_choi}) , the last statement is equivalent to Equation~\ref{apeendixequationstep2}: we find $\Lambda_{B_1}\otimes\Lambda_{B_2} (B^b_{B_1B_2})=4\Trr{B}{B^b_{B_1B_2}\sigma_{A'B_1}\otimes\sigma_{B_2C'}}=\Phi^b_{A'C'}$. 
\end{proof}

\section*{Appendix C: Generalization}

This result of Theorem 1 can be generalized to other entangled measurements. A common way to self-test a state is to construct extraction channels (here $\Gamma^0_A, \Gamma^0_C$) out of the local measurement operators.
As we show in the following, it is often possible self-test a family of states which form a basis of the considered Hilbert space by relabeling the measurement operators. If the corresponding extraction channels are all linked together in a specific way, the proof of Theorem 1 can be generalized to self-test the measurement in the corresponding basis. 
Here, we show explicitly that this is the case for the tilted BSM and the GHZ measurement.

\section{Tilted Bell-state measurement}

The proof of our main result can be extended to the case of the tilted Bell-state measurement (tilted BSM), which is a measurement in the basis:
\begin{align*}
&\ket{\phi^+_\theta}=c_\theta\ket{00}+s_\theta\ket{11},
&\ket{\phi^-_\theta}=s_\theta\ket{00}-c_\theta\ket{11},\\
&\ket{\psi^+_\theta}=c_\theta\ket{01}+s_\theta\ket{10},
&\ket{\psi^-_\theta}=s_\theta\ket{01}-c_\theta\ket{10},
\end{align*}
where $0<\theta\leq\pi/4$ and $c_\theta=\cos{\theta}$, $s_\theta=\sin{\theta}$. Remark that for $\theta=\pi/4$, we recover the usual Bell 
states. We call them $\Phi_\theta^b$, for $b=0,\cdots, 3$, keeping the same ordering.
We consider again an entanglement swapping scenario in which Alice/Bob and Bob/Charlie share a maximally entangled state $\phi^+$, Bob now performs the tilted BSM. Alice and Charlie perform the ideal local measurements 
$A'_{0} := \sigma_{z}$, $A'_{1} := \sigma_{x}$, $C'_{0} := ( \cos\mu\sigma_{z} + \sin\mu\sigma_{x} ) / \sqrt{2}$, $C'_{1} := ( \cos\mu \sigma_{z} - \sin\mu\sigma_{x} ) / \sqrt{2} $ where $\tan(\mu)=\sin(2\theta)$. Then, their statistics conditioned on $b$ will maximally violate a version of the tilted CHSH inequality.
For $b=0$, we have that 
\begin{align*}
\CHSH_0^\eta&:=+\eta \mean{A_0}\\+&\mean{A_0 C_0}+\mean{A_0 C_1}+\mean{A_1 C_0}-\mean{A_1 C_1},
\end{align*}
where $\eta=2/\sqrt{1+2\mathrm{tan}^2(2\theta)}$. The other variants are $\CHSH_b^\eta$, obtained with the symmetries introduced for the $\CHSH$ case: $\CHSH_1^\eta$ is obtained from $\CHSH_0^\eta$ with $A_1\rightarrow-A_1$, $\CHSH_2^\eta:=-\CHSH_1^\eta$ and $\CHSH_3^\eta:=-\CHSH_0^\eta$.

The formulation of the self-testing result and the proof can directly be deduced from the $\CHSH$ case, where the anticommuting operators are defined in \cite{bamps}.

\section{GHZ measurement}

The GHZ measurement features 8 eigenstates, given by the eight GHZ states
\begin{align*}
&\ket{\GHZ^{0,\pm}}=\frac{\ket{000}\pm\ket{111}}{\sqrt{2}},
&\ket{\GHZ^{A,\pm}}=\frac{\ket{011}\pm\ket{100}}{\sqrt{2}},\\
&\ket{\GHZ^{B,\pm}}=\frac{\ket{101}\pm\ket{010}}{\sqrt{2}},
&\ket{\GHZ^{C,\pm}}=\frac{\ket{110}\pm\ket{001}}{\sqrt{2}}.
\end{align*}
Note that here we use for convenience the labels $0,A,B,C$. In the ideal protocol, Alice, Bob and Charlie independently share maximally entangled states $\ket{\phi^+}$ with a central party Rob and measure ${A'}_0={B'}_0={C'}_0={X'}$ and ${A'}_1={B'}_1={C'}_1={Y'}$. Rob measures in the $\GHZ$ basis. The considered Bell expression is the Mermin Inequality, of maximal violation 4, give by
\begin{equation}
\Mer^{0,+}:=\mean{A_0 B_0 C_0}-\mean{A_0 B_1 C_1}-\mean{A_1 B_0 C_1}-\mean{A_1 B_1 C_0}.
\end{equation}
The other used symmetries of the Mermin inequality are $\Mer^r$ for $r=(P,\pm)$ with $P\in \{0, A, B, C\}$. More precisely, $\Mer^{A,+}$ is obtained from $\Mer^{0,+}$ with $A_1\rightarrow -A_1$ and similarly for $\Mer^{B,+}$, $\Mer^{C,+}$, and $\Mer^{P,-}:=-\Mer^{P,+}$.
In this ideal scenario, there is a maximal violation of the inequality $\Mer^r$ conditioned on Rob result $r=(P,\pm)$.
We first introduce the formal Pauli matrices

\begin{align*}
X_A&=A_0, &X_B&=B_0, &X_C&=C_0,\\
Y_A&=A_1, &Y_B&=B_1, &Y_C&=C_1,\\
Z_A&=-i X_A Y_A, &Z_B&=-i X_B Y_B, &Z_C&=-i X_C Y_C.
\end{align*}
We have the following theorem:
\setcounter{step}{0}
\setcounter{theorem}{2}
\begin{theorem}\label{theoremGHZselftest}
Let the initial state shared by Alice, Bob, Charlie and Rob be of the form
\begin{equation*}
\tau = \tau_{A R_{A}} \otimes \tau_{B R_{B}} \otimes \tau_{C R_{C}}
\end{equation*}
and let $\cR := ( R_{R_{A}R_{B}R_{C}}^r)_{r=(P,\pm)}$ be an eight-outcome measurement acting on $\cH_{R_{A}} \otimes \cH_{R_{B}} \otimes \cH_{R_{C}}$.
If there exist measurements for Alice, Bob, Charlie and Rob such that the resulting statistics conditioned on $r$ exhibit the maximal violation of the $\Mer^{r}$ inequality, then there exist completely positive and unital maps $\Lambda_{R_{A}} : \cL( \cH_{ R_{A} } ) \to \cL( \cH_{A'} ), \Lambda_{R_{B}} : \cL( \cH_{ R_{B} } ) \to \cL( \cH_{B'} ), \Lambda_{R_{C}} : \cL( \cH_{ R_{C} } ) \to \cL( \cH_{C'} )$ for $\abs{ A' } = \abs{ B' } = \abs{ C' } = 2$ such that
\begin{equation*}
\big( \Lambda_{R_{A}} \otimes \Lambda_{R_{B}} \otimes \Lambda_{R_{C}} \big) \big( R_{ R_{A} R_{B} R_{C} }^{r}\big) =\GHZ_{A'B'C'}^{r}
\end{equation*}
for $r=(P,\pm)$ with $P \in \{0, A, B, C\}$.
\end{theorem}
The proof is in similar to the previous one, in two steps.
We introduce the notation $r_0=(+,0)$.
\begin{step}\label{Step1GHZ}
For $P=A,B,C$, let $\Gamma^{r_0}_P=\Gamma_{X_P,Z_P}$. We introduce the reduced states
\begin{align*}
\sigma_{PR_P}:&=\Gamma^{r_0}_{P}(\tau_{PR_P}),\\
\sigma^{r}_{A'B'C'}:&=\left(\Gamma^{r_0}_{A}\otimes\Gamma^{r_0}_{B}\otimes \Gamma^{r_0}_{C}\right)(\tau^{r}_{AC})\\&=8\Tr{R^{r}_{R_AR_BR_C} \left( \sigma_{A'R_A}\otimes\sigma_{B'R_B}\otimes\sigma_{C'R_C}\right)},
\end{align*}
where $\tau^{r}_{AC}$ is the state shared between Alice, Bob and Charlie after Rob measured $r$.
Then, we have:
\begin{equation}\label{apeendixequationstep2_GHZ}
\sigma^{r}_{A'B'C'}=\GHZ^{r}_{A'B'C'}
\end{equation}
\end{step}
\begin{proof}
For result $r=r_0$, the proof is similar to the $\CHSH$ case. Considering the square of the Bell operator associated to $\Mer^{r_0}$, on can show that for any party $P$, the formal Pauli matrices $X_P, Y_P$ anti-commutes and square to identity over the support of $\tau$. This implies that $X_P, Z_P$ anti-commutes and square to identity over the support of $\tau$. Moreover, in the rest of the derivation, $Y_P$ can always be replaced with $-iZ_PX_P$ over the support of $\tau$. Then, the maximal violation of $\Mer^{r_0}$ can be used to prove that $\left(S^{r_0}_A\otimes S^{r_0}_B\otimes S^{r_0}_C\right)(\ketbra{000}{000}\otimes\tau^0_{ABC})$ is a product state between $\cH_{ABC}:=\cH_{A}\otimes \cH_{B}\otimes\cH_{C}$ and $\cH_{A'B'C'}:=\cH_{A'}\otimes \cH_{B'}\otimes\cH_{C'}$, and 
\begin{equation*}
\left(\Gamma^{r_0}_{A}\otimes\Gamma^{r_0}_{B}\otimes \Gamma^{r_0}_{C}\right)(\tau^{r_0}_{ABC})=\GHZ^{r_0}_{A'B'C'}.
\end{equation*}

Let us now prove Step~\ref{Step1GHZ} for any $r=(P,\pm)$. For $P=0,A,B,C$, we introduce the operators 
\begin{align*}
T^{P,+}&=X_P, &T^{P,-}&=T^{P,+}Z_AZ_BZ_C,\\ {T'}^{P,+}&=X'_P, &{T'}^{P,-}&={T'}^{P,+}Z'_AZ'_BZ'_C,
\end{align*}
where $T^{r_0}=\unit$. 
A straightforward calculation shows that for $x=0,1$, $\Mer^r$ is formally linked to $\Mer^{r_0}$ by the transformation $A_x\rightarrow T^r A_x T^r=\epsilon^A_{r,x} A_x$ with the anti-commutation rule $X_AZ_A=-Z_AX_A$, which define a sign $\epsilon^A_{r,x}=\pm 1$, and similarly transformation for $B, C$ (which define $\epsilon^B_{r,y}, \epsilon^C_{r,z}=\pm 1$ for $y,z=0,1$).
Hence, considering 
\begin{align*}
S^{r}_P&=S_{\epsilon^P_{r,t}X_P,\epsilon^P_{r,t}Z_P}, &\Gamma^{r}_P&=\Gamma_{\epsilon^P_{r,t}X_P,\epsilon^P_{r,t}Z_P}
\end{align*}
for $P=0,A,B,C$ and $t=0,1$, we can exploit the proof for $r=(0,+)$ to obtain 
\begin{equation*}
\left(\Gamma^r_{A}\otimes\Gamma^r_{B}\otimes \Gamma^r_{C}\right)(\tau^r_{ABC})=\GHZ^{r_0}_{A'B'C'}.
\end{equation*}
Then with Lemma~\ref{LemmaSwap}, basic computations similar to the $\CHSH$ case show that 
\begin{align*}
\left(\Gamma^{r_0}_{A}\otimes\Gamma^{r_0}_{B}\otimes \Gamma^{r_0}_{C}\right)(\tau^r_{ABC})&=T'^r\left(\Gamma^r_{A}\otimes\Gamma^r_{B}\otimes \Gamma^r_{C}\right)(\tau^r_{ABC}){T'}^r\\&=T'^r\GHZ^{r_0}_{A'B'C'}T'^r\\&=\GHZ^r_{A'B'C'}.
\end{align*}
\end{proof}

\begin{step}
For $P\in\{A,B,C\}$, let $\Lambda_{R_P}:\cL(\cH_{R_P})\rightarrow\cL(\cH_{P'})$ be the Choi-Jamio{\l}kowski map associated to the operator $2\sigma_{PR_P}$. This map is unital and
\begin{equation}
\Lambda_{R_A}\otimes\Lambda_{R_B}\otimes\Lambda_{R_C}(R^r_{R_AR_BR_C})=\GHZ_{A'B'C'}^{r}
\end{equation}
for $r=(P,\pm)$ with $P \in \{0, A, B, C\}$, which proof Theorem~\ref{theoremGHZselftest}.
\end{step}
\begin{proof}
The proof is exactly similar to the one of Step~\ref{step2_CHSH}. of Theorem~\ref{thm:self-testing-BSM-exact}. We first prove unitality and then show that the final statement is no more than a rewriting of Eq.~(\ref{apeendixequationstep2_GHZ}).
\end{proof}

\section{Further generalization}

The two previous examples demonstrate that the method used to self-test the BSM can be generalized to other entangled measurements on qubits. We expect that our method directly generalizes to a basis created out of $N$-qubit or higher dimensional systems.
Our proof relies on two steps. First, a self-test result of a state $\ket{\Phi^0}$ in which a tensor product of local extraction map $\Gamma^0$ is constructed out of the measurement operators of the parties.
Second, symmetries which allow to self-test a full basis of state $\ket{\Phi^r}$ with the same measurement operators and such that the corresponding extraction maps $\Gamma^r$ can be linked to $\Gamma^0$ in a proper way.

We expect that recent work developed methods for self-testing different classes of states, such as \cite{supic}, can be exploited to generalize our result.

\ifthenelse{\equal{\shortversion}{1}}
{\end{comment}}{}

\section*{Appendix D: Basic properties of the quality measure}
\label{app:quality-measure}
In this appendix we prove some basic properties of the quality measure defined in the main text and for completeness let us first restate the definition. We consider two measurements with $d$ outcomes: the ideal, projective measurement $\sP = ( P_{A'}^{j} )_{j = 1}^{d}$ acting on $\cH_{A'}
$ and the real (not necessarily projective) measurement $\sF = ( F_{A}^{j} )_{j = 1}^{d}$ acting on $\cH_{A}$. The quality of $\sF$ as a simulation of $\sP$ is given by
\begin{equation}
\label{eq:q-unipartite-definition}
\Qunipartite,
\end{equation}
where $\abs{A'}$ is the dimension of the Hilbert space $\cH_{A'}
$ and the maximization is taken over completely positive unital maps $\Lambda : \cL( \cH_{A} ) \to \cL( \cH_{A'}
 )$. Let us start by showing that $Q( \sF, \sP ) \in [0, 1]$ and examining the extremal cases.
\begin{prop}
We always have
\begin{equation*}
Q( \sF, \sP ) \geq \frac{1}{ \abs{A'} \cdot \abs{A} } \sum_{j = 1}^{d} \trj \big( F_{A}^{j} \big) \cdot \trj \big( P_{A'}^{j} \big).
\end{equation*}
Moreover, if the right-hand side vanishes, we must have $Q( \sF, \sP ) = 0$.
\end{prop}
\begin{proof}
The bound comes from the map
\begin{equation*}
\Lambda(X) := \frac{\Tr X}{ \abs{A} } \cdot \unit_{A'},
\end{equation*}
which is easily checked to be completely positive and unital. If the right-hand side vanishes, we must have
\begin{equation*}
\trj \big( F_{A}^{j} \big) \cdot \trj \big( P_{A'}^{j} \big) = 0
\end{equation*}
for every $j$, which implies that for every $j$ either $F_{A}^{j} = 0$ or $P_{A'}^{j} = 0$. In such a case every term on the right-hand side of Eq.~\eqref{eq:q-unipartite-definition} must vanish regardless of the choice of $\Lambda$.
\end{proof}
Note that if $\sP$ is not only projective, but also rank-1 (which implies $d = \abs{A'}$), this lower bound simplifies to $Q(\sF, \sP) \geq \frac{1}{d}$.
\begin{prop}
\label{prop:Q-upper-bound}
We always have $Q( \sF, \sP ) \leq 1$. Moreover, if $Q( \sF, \sP ) = 1$, then there exists a completely positive unital map $\Lambda$ such that
\begin{equation*}
\Lambda ( F_{A}^{j} ) = P_{A'}^{j}
\end{equation*}
for every $j$.
\end{prop}
\begin{proof}
For a fixed map $\Lambda$ let $Q_{A'}^{j} := \Lambda( F_{A}^{j} )$. Since the map is completely positive we have $Q_{A'}^{j} \geq 0$ for all $j$ and since it is unital we have
\begin{equation*}
\sum_{j = 1}^{d} Q_{A'}^{j} = \sum_{j = 1}^{d} \Lambda( F_{A}^{j} ) = \Lambda(\unit_{A}) = \unit_{A'},
\end{equation*}
which together implies that $Q_{A'}^{j} \leq \unit_{A'}$ for every $j$. Therefore,
\begin{equation}
\label{eq:Q-upper-bound}
\frac{1}{\abs{A'}} \sum_{j = 1}^{d} \ave[\big]{ Q_{A'}^{j}, P_{A'}^{j} } \leq \frac{1}{\abs{A'}} \sum_{j = 1}^{d} \ave[\big]{ \unit_{A'}, P_{A'}^{j} }
= 1.
\end{equation}
Since this bound holds for every completely positive unital map $\Lambda$, we immediately obtain $Q( \sF, \sP ) \leq 1$.

If $Q(\sF, \sP) = 1$, there exists a map $\Lambda$ such that the resulting operators $Q_{A'}^{j}$ saturate the upper bound given in Eq.~\eqref{eq:Q-upper-bound}. This means that the equality
\begin{equation*}
\ave[\big]{ Q_{A'}^{j}, P_{A'}^{j} } = \ave[\big]{ \unit_{A'}, P_{A'}^{j} }
\end{equation*}
holds for every $j$, which implies that $Q_{A'}^{j} \geq P_{A'}^{j}$ for all $j$. Finally, the relation
\begin{equation*}
\unit_{A'} = \sum_{j = 1}^{d} Q_{A'}^{j} \geq \sum_{j = 1}^{d} P_{A'}^{j} = \unit_{A'}
\end{equation*}
forces all these inequalities to hold as equalities.
%
\end{proof}
If $\sF$ and $\sP$ act jointly on two subsystems the quality of simulation is given by
\begin{align*}
Q( \sF, \sP ) &:= \frac{1}{\abs{B_{1}'} \cdot \abs{B_{2}'}}\\
&\max_{ \Lambda_{B_{1}}, \Lambda_{B_{2}} } \sum_{j = 1}^{d} \ave[\big]{ ( \Lambda_{B_{1}} \otimes \Lambda_{B_{2}} ) ( F_{ B_{1} B_{2} }^{j} ), P_{ B_{1}' B_{2}' }^{j} }
\end{align*}
where the maximization is taken over completely positive unital maps $\Lambda_{X} : \cL( \cH_{X'} ) \to \cL( \cH_{X} )$ for $X = B_{1}, B_{2}$. Proposition~\ref{prop:Q-upper-bound} straightforwardly extends to the bipartite setting, so let us state it without a proof.
\begin{prop}
\label{prop:Q-upper-bound-bipartite}
We always have $Q( \sF, \sP ) \leq 1$. Moreover, if $Q( \sF, \sP ) = 1$, then there exist completely positive unital maps $\Lambda_{B_{1}} : \cL( \cH_{B_{1}} ) \to \cL( \cH_{B_{1}'} )$ and $\Lambda_{B_{2}} : \cL( \cH_{B_{2}} ) \to \cL( \cH_{B_{2}'} )$ such that
\begin{equation*}
\exdefbi
\end{equation*}
for every $j$.
\end{prop}
The separability threshold is given by
\begin{equation*}
\Qsepdef,
\end{equation*}
where $\cM_{\textnormal{sep}}$ is the set of separable measurements acting on $\cH_{B_{1}'} \otimes \cH_{B_{2}'}$.

Proposition~\ref{prop:Q-upper-bound-bipartite} implies that if $\sP$ is an entangled measurement, we must have $\Qsep(\sP) < 1$, but computing an explicit upper bound is not entirely trivial. In the following proposition we compute an explicit upper bound for measurements composed of rank-1 projectors.
\begin{prop}
\label{prop:rank1-projective-upper-bound}
Let $\sP$ be a rank-1 projective measurement given by
\begin{equation*}
P_{B_{1}'B_{2}'}^{j} = \ketbraq{e_{j}}_{B_{1}'B_{2}'}
\end{equation*}
and let the Schmidt decomposition of $\ket{e_{j}}_{B_{1}'B_{2}'}$ be
\begin{equation*}
\ket{e_{j}}_{B_{1}'B_{2}'} = \sum_{l} \alpha_{j, l} \ket{a_{j, l}}_{B_{1}'} \ket{b_{j, l}}_{B_{2}'}.
\end{equation*}
Then,
\begin{equation*}
\Qsep(\sP) \leq \alpha_{\textnormal{max}}^{2},
\end{equation*}
where $\alpha_{\textnormal{max}} := \max_{j,l} \alpha_{j, l}$ is the largest Schmidt coefficient.
\end{prop}
\begin{proof}
For fixed channels $\Lambda_{A}$ and $\Lambda_{B}$ let $Q_{B_{1}'B_{2}'}^{j} := ( \Lambda_{A} \otimes \Lambda_{B} )(F_{AB}^{j})$. These are still separable operators, i.e.~we can write them as
\begin{equation*}
Q_{B_{1}'B_{2}'}^{j} = \sum_{k} \lambda_{j, k} \ketbraq{\psi_{j, k}}_{B_{1}'B_{2}'}
\end{equation*}
for some product states $\ket{\psi_{j, k}}$. Moreover, they satisfy $\sum_{j = 1}^{d} Q_{B_{1}'B_{2}'}^{j} = \unit_{B_{1}'B_{2}'}$. Then, we obtain
\begin{align*}
\sum_{j = 1}^{d} \ave[\big]{ Q_{B_{1}'B_{2}'}^{j}, P_{B_{1}'B_{2}'}^{j} } &= \sum_{j = 1}^{d} \sum_{k} \lambda_{j, k} \abs{\braket{e_{j}}{\psi_{j, k}}}^{2}\\
&\leq \alpha_{\textnormal{max}}^{2} \sum_{j = 1}^{d} \sum_{k} \lambda_{j, k}\\
&= \alpha_{\textnormal{max}}^{2} \sum_{j = 1}^{d} \trj \big( Q_{B_{1}'B_{2}'}^{j} \big)\\
&= \alpha_{\textnormal{max}}^{2} \cdot \trj ( \unit_{B_{1}'B_{2}'} )\\
&= \alpha_{\textnormal{max}}^{2} \cdot \abs{B_{1}'} \cdot \abs{B_{2}'},
\end{align*}
where we have used the fact that the overlap between a product state and an entangled state cannot exceed the square of the largest Schmidt coefficient. Since the bound does not depend on the specific choice of maps, it holds universally.
\end{proof}
Clearly, the estimate above is rather crude, but it can be tight, e.g.~for the BSM we obtain the value of $\frac{1}{2}$ which turns out to be correct.

A tighter bound can be obtained if we take into account the individual Schmidt coefficients of the ideal projectors. Clearly,
\begin{align*}
\sum_{j = 1}^{d} \ave[\big]{ Q_{B_{1}'B_{2}'}^{j}, P_{B_{1}'B_{2}'}^{j} } &\leq \sum_{j = 1}^{d} \alpha_{j, \textnormal{max}}^{2} \sum_{k} \lambda_{j, k}\\
&= \sum_{j = 1}^{d} \alpha_{j, \textnormal{max}}^{2} \trj \big( Q_{B_{1}'B_{2}'}^{j} \big),
\end{align*}
where $\alpha_{j, \textnormal{max}} := \max_{l} \alpha_{j, l}$. In the last step we must determine the choice of traces $\trj \big( Q_{B_{1}'B_{2}'}^{j} \big)$ that maximizes the right-hand side of this expression. Let us order the outcomes such that the coefficients $\alpha_{j, \textnormal{max}}$ are non-increasing. The bound stated in Proposition~\ref{prop:rank1-projective-upper-bound} corresponds to assigning all the trace to $j = 1$. However, since each individual term is upper-bounded by $\trj P_{B_{1}'B_{2}'}^{j} = 1$, it suffices to set $\trj Q_{B_{1}'B_{2}'}^{1} = \alpha_{1, \textnormal{max}}^{-2}$ and then distribute the remaining trace over the other terms. For $j \geq 2$ the optimal choice is given by
\begin{equation*}
\trj Q_{B_{1}'B_{2}'}^{j} = \min \bigg\{ \alpha_{j, \textnormal{max}}^{-2}, \abs{B_{1}'} \cdot \abs{B_{2}'} - \sum_{k = 1}^{j - 1} \alpha_{k, \textnormal{max}}^{-2} \bigg\}.
\end{equation*}
It is easy to verify that the resulting upper bound is non-trivial as long as there are some entangled projectors and in some cases it can even be tight. 
If we choose a measurement composed of two product states and two Bell states $( \ketbraq{00}, \ketbraq{11}, \Phi^{2}, \Phi^{3} ) $, we obtain the value of $\frac{3}{4}$ which turns out to be tight.

The measure given in Eq.~\eqref{eq:q-unipartite-definition} captures how well the real measurement $\sF$ simulates the ideal measurement $\sP$. Since the measure is simply a sum over terms corresponding to all possible measurement outcomes, one might be tempted to think that in order to certify the quality of a \emph{single measurement operator}, it would suffice to look at the relevant term. 
This is, however, not quite true as shown by the following example. Suppose we want to certify that $F_{A}^{0}$ is capable of simulating a rank-1 projector $P_{A'}^{0}$.
While the upper bound
\begin{equation*}
\ave[\big]{ \Lambda ( F_{A}^{0} ), P_{A'}^{0} } \leq \ave[\big]{ \unit_{A'}, P_{A'}^{0} } = 1
\end{equation*}
still holds, saturating it does \emph{not} allow us to conclude that $\Lambda( F_{A}^{0} ) = P_{A'}^{0}$. In particular, another valid solution is given by $\Lambda( F_{A}^{0} ) = \unit_{A'}$. In order to construct a measure which is maximized iff $\Lambda( F_{A}^{0} ) = P_{A'}^{0}$ one must include an extra component, e.g.~the trace of the resulting operator. Indeed, the conditions
\begin{equation*}
\ave[\big]{ \Lambda ( F_{A}^{0} ), P_{A'}^{0} } = 1 \nbox{and} \trj \big( \Lambda ( F_{A}^{0} ) \big) = 1
\end{equation*}
are sufficient to conclude $\Lambda( F_{A}^{0} ) = P_{A'}^{0}$. In particular, this means that in the bipartite case the entanglement of a single measurement operator \emph{cannot} be inferred by looking only at $\ave{ \Lambda ( F_{B_{1} B_{2}}^{0} ), P_{B_{1}'B_{2}'}^{0} }$. For instance, if $P_{B_{1}'B_{2}'}^{0}$ is a rank-1 entangled projector, the maximal value of
\begin{equation*}
\ave[\big]{ \Lambda ( F_{B_{1} B_{2}}^{0} ), P_{B_{1}'B_{2}'}^{0} } = 1
\end{equation*}
can be achieved by a separable measurement operator $F_{B_{1} B_{2}}^{0}$, e.g.~$F_{B_{1} B_{2}}^{0} = \unit_{B_{1} B_{2}}$.
\section*{Appendix E: Noise tolerant results}
\section{Proof of Theorem~2}
\label{app:BSM-robust}
In this appendix we provide a complete proof of Theorem~2. We begin by proving three auxiliary lemmas. The first one concerns an arbitrary two-qubit state. For a Hermitian operator $X$ we denote its spectrum by $\spec(X)$.
\begin{lemma}
\label{lem:fidelity-marginal-tradeoff}
Let $\rho_{AB}$ be a two-qubit state and let $\Phi_{AB}$ some pure maximally entangled state. If $F( \rho_{AB}, \Phi_{AB} ) \geq c$ for some $c \in [\frac{1}{2}, 1]$, then
\begin{equation*}
\spec ( \rho_{A} ) \subseteq \bigg[ \frac{1 - \eta}{2}, \frac{1 + \eta}{2} \bigg]
\end{equation*}
for $\eta := 2 \sqrt{ c (1 - c) }$.
\end{lemma}
\begin{proof}
Note that if $c = 1$, then $\rho_{AB} = \Phi_{AB}$ and we necessarily have $\spec( \rho_{A} ) = \{ \frac{1}{2} \}$. For $c \in [\frac{1}{2}, 1)$ we find the trade-off by solving a semidefinite program in which we constrain the fidelity with the maximally entangled state and maximize the expectation value of some single-qubit Pauli observable. This gives the upper bound on the spectrum of $\rho_{A}$, whereas the lower bound follows from normalization. Without loss of generality we can assume the maximally entangled state to be $\ket{\Phi^{+}} = \frac{1}{\sqrt{2}} ( \ket{00} + \ket{11} )$ and the Pauli observable to be $\sigma_{z}$. Then, the primal problem reads
\begin{equation*}
\begin{array}{ll}
\nbox{maximize} &\ave{ \sigma_{z}, \rho_{A} }\\
\nbox{subject to} &\ave{ \Phi_{AB}^{+} , \rho_{AB} } \geq c,\\
&\ave{ \unit, \rho_{AB} } = 1,\\
\nbox{over} &\rho_{AB} \geq 0.
\end{array}
\end{equation*}
Computing the dual leads to
\begin{equation*}
\begin{array}{ll}
\nbox{minimize} & \lambda_{1} - \lambda_{2} c\\
\nbox{subject to} &\lambda_{1} \unit \geq \lambda_{2} \Phi^{+}_{AB} + \sigma_{z} \otimes \unit,\\
\nbox{over} &\lambda_{1} \in \amsbb{R},\\
&\lambda_{2} \geq 0.
\end{array}
\end{equation*}
For $c \in [1/2, 1)$ the assignment
\begin{equation*}
\lambda_{1} = \sqrt{ \frac{ c }{ 1 - c } } \nbox{and} \lambda_{2} = \frac{ 2c - 1 }{ \sqrt{ c ( 1 - c ) } }
\end{equation*}
constitutes a valid solution to the dual and the corresponding value equals $2 \sqrt{ c ( 1 - c ) }$.
\end{proof}
To see that this bound cannot be improved, note that it is saturated by pure partially entangled states $\ket{\psi^{\theta}} := \cos \theta \ket{00} + \sin \theta \ket{11}$ for $\theta \in [0, \pi/4]$ (the fidelity equals $c = (1 + \sin 2\theta )/2$, the spectrum of the reduced state is $\{\cos^{2} \theta, \sin^{2} \theta\}$ ).

In the second lemma we prove an operator inequality for an arbitrary qubit-qudit state.
\begin{lemma}
\label{lem:rescaling-marginal}
Let $\nu_{AB}$ be a qubit-qudit state such that
\begin{equation*}
\spec( \nu_{A} ) = \bigg\{ \frac{ 1 - \eta }{2}, \frac{ 1 + \eta }{2} \bigg\}
\end{equation*}
for some $\eta \in [0, 1)$. Moreover, define
\begin{equation*}
\mu_{AB} := \sqroot{\nu}{}.
\end{equation*}
Then, the operator inequality
\begin{equation}
\label{eq:mu-operator-inequality}
\mu_{AB} \geq s( \eta ) \nu_{AB} - t(\eta) \frac{ \unit }{2} \otimes \nu_{B}
\end{equation}
holds for
\begin{equation*}
s( \eta ) := \frac{2}{ \sqrt{ 1 - \eta^{2} } } \nbox{and} t( \eta ) := \frac{4}{ \sqrt{ 1 - \eta^{2} } } - \frac{4}{ 1 + \eta }.
\end{equation*}
\end{lemma}
\begin{proof}
An arbitrary qubit-qudit state can be written as
\begin{equation}
\nu_{AB} = \frac{1}{2} \big( \unit \otimes E_{0} + \sum_{j} \sigma_{j} \otimes E_{j} \big)
\end{equation}
for some Hermitian operators $E_{j}$ (the summation goes over $j \in \{x, y, z\}$). Since $\nu_{B} = E_{0}$, we must have $E_{0} \geq 0$ and $\trj ( E_{0} ) = 1$. Moreover, since the operator
\begin{align*}
\nu_{AB} &+ (\sigma_{z} \otimes \unit) \nu_{AB} (\sigma_{z} \otimes \unit)\\
&= \unit \otimes E_{0} + \sigma_{z} \otimes E_{z}\\
&= \ketbraq{0} \otimes ( E_{0} + E_{z}) + \ketbraq{1} \otimes ( E_{0} - E_{z})
\end{align*}
is positive semidefinite, we also have $E_{0} + E_{z} \geq 0$ and $E_{0} - E_{z} \geq 0$. Without loss of generality we can assume that the reduced state $\nu_{A}$ is diagonal in the computational basis, i.e.
\begin{equation*}
\nu_{A} = \frac{ \unit + \eta \sigma_{z} }{2}.
\end{equation*}
This implies that $\trj (E_{x}) = \trj (E_{y}) = 0$ and $\trj (E_{z}) = \eta$. Since $\eta \in [0, 1)$, the reduced state is full-rank and the inverse is well-defined. The square root of the inverse is given by
%
%
%
%
\begin{equation*}
\nu_{A}^{-1/2} = a(\eta) \unit - b(\eta) \sigma_{z}
\end{equation*}
for
\begin{align*}
a(\eta) := \frac{ \sqrt{ 1 + \eta} + \sqrt{ 1 - \eta } }{ \sqrt{ 2 ( 1 - \eta^{2} ) } } \nbox{and} b(\eta) := \frac{ \sqrt{ 1 + \eta} - \sqrt{ 1 - \eta } }{ \sqrt{ 2 ( 1 - \eta^{2} ) } }.
\end{align*}
\begin{widetext}
\noindent Computing $\mu_{AB}$ gives
\begin{align*}
\mu_{AB} =& \big( \nu_{A}^{-1/2} \otimes \unit \big) \nu_{AB} \big( \nu_{A}^{-1/2} \otimes \unit \big)\\
=& \frac{1}{2} \bigg( \nu_{A}^{-1} \otimes E_{0} + \frac{2}{ \sqrt{ 1 - \eta^{2} } } \Big[ \sigma_{x} \otimes E_{x} + \sigma_{y} \otimes E_{y} \Big] + \frac{2}{ 1 - \eta^{2} } ( - \eta \unit + \sigma_{z} ) \otimes E_{z} \bigg)\\
=& \unit \otimes \frac{ E_{0} - \eta E_{z} }{ 1 - \eta^{2} } + \frac{1}{ \sqrt{ 1 - \eta^{2} } } \Big[ \sigma_{x} \otimes E_{x} + \sigma_{y} \otimes E_{y} \Big] + \sigma_{z} \otimes \frac{ - \eta E_{0} + E_{z} }{ 1 - \eta^{2} }.
\end{align*}
The operator inequality~\eqref{eq:mu-operator-inequality} is equivalent to
\begin{equation*}
\mu_{AB} - s( \eta ) \nu_{AB} + t(\eta) \frac{ \unit }{2} \otimes \nu_{B} \geq 0.
\end{equation*}
Writing out the left-hand side gives
\begin{align*}
\mu_{AB} - s( \eta ) \nu_{AB} + t(\eta) \frac{ \unit }{2} \otimes \nu_{B} &= \unit \otimes \bigg( \frac{ E_{0} - \eta E_{z} }{ 1 - \eta^{2} } + \frac{ \big[ - s(\eta) + t(\eta) \big] E_{0} }{ 2 } \bigg) + \sigma_{z} \otimes \bigg( \frac{ - \eta E_{0} + E_{z} }{ 1 - \eta^{2} } - \frac{ s(\eta) E_{z} }{ 2 } \bigg)\\
&= \ketbraq{0} \otimes \bigg[ \bigg( \frac{ 1 }{ 1 + \eta } - \frac{ s(\eta) }{2} \bigg) ( E_{0} + E_{z} ) + \frac{ t(\eta) E_{0} }{2} \bigg]\\
&+ \ketbraq{1} \otimes \bigg[ \bigg( \frac{ 1 }{ 1 - \eta } - \frac{ s(\eta) }{2} \bigg) ( E_{0} - E_{z} ) + \frac{ t(\eta) E_{0} }{2} \bigg].
\end{align*}
To show positivity it suffices to analyse each block separately. Positivity of the $\ketbraq{1}$ block is clear (it is a sum of two positive semidefinite operators), but the $\ketbraq{0}$ block requires more work. Since for $\eta \in [0, 1)$ we have
\begin{equation*}
\frac{ 1 }{ 1 + \eta } - \frac{ s(\eta) }{ 2 } \leq 0,
\end{equation*}
we apply the bound $E_{z} \leq E_{0}$ to obtain
\begin{align*}
\bigg( \frac{ 1 }{ 1 + \eta } - \frac{ s(\eta) }{2} \bigg) ( E_{0} + E_{z} ) + \frac{ t(\eta) E_{0} }{2} \geq \bigg( \frac{ 2 }{ 1 + \eta } - s(\eta) + \frac{ t(\eta) }{2} \bigg) E_{0} = 0.
\end{align*}
\end{widetext}
\end{proof}
The last lemma is a simple generalization of the CHSH self-testing result from Ref.~\cite{kaniewski16b}, which shows that the same extraction channels can be used for different variants of the CHSH inequality.
\begin{lemma}
\label{lem:simultaneous-self-testing}
For $b \in \{0, 1, 2, 3\}$ let $\tau_{AC}^b$ be arbitrary normalized states acting on $\cH_{A} \otimes \cH_{C}$. For observables $A_{0}, A_{1}$ acting on $\cH_{A}$ and $C_{0}, C_{1}$ acting on $\cH_{C}$ define the Bell operators
\begin{align*}
W_{0} &:= A_{0} \otimes C_{0} + A_{0} \otimes C_{1} + A_{1} \otimes C_{0} - A_{1} \otimes C_{1},\\
W_{1} &:= A_{0} \otimes C_{0} + A_{0} \otimes C_{1} - A_{1} \otimes C_{0} + A_{1} \otimes C_{1},\\
W_{2} &:= - W_{1},\\
W_{3} &:= - W_{0}.
\end{align*}
and let $\beta_{b} := \trj (W_{b} \tau_{AC}^b)$ be the corresponding Bell value. Then, there exist quantum channels $\Gamma_{A} : \cL( \cH_{A} ) \to \cL( \cH_{A'} )$ and $\Gamma_{C} : \cL( \cH_{C} ) \to \cL( \cH_{C'} )$, where $\abs{ A' } = \abs{ C' } = 2$, such that for $b \in \{0, 1, 2, 3\}$
\extractabilitylb{,}
where
\gfunc{}
for $x^{*} := ( 16 + 14 \sqrt{2} ) / 17$.
\end{lemma}
\begin{proof}
The original result proves only the statement corresponding to $b = 0$. However, since the extraction channels depend only on the observables, one might expect that the same choice works equally well for other variants of the CHSH inequality. Indeed, if we keep precisely the same extraction channels and write down the operator inequalities corresponding to $b = 1, 2, 3$, we realize they are all unitarily equivalent to the $b = 0$ case, which leads to analogous self-testing statements.
\end{proof}
We are now ready to prove the main theorem.
\setcounter{theorem}{1}
\begin{theorem}\label{thm:self-testing-BSM-robust}
\BSMtheorem{1}
\end{theorem}
\begin{proof}
Recall that $\tau_{AC}^b$ is defined as
\taubdef
Lemma~\ref{lem:simultaneous-self-testing} guarantees the existence of completely positive trace-preserving maps $\Gamma_{A} : \cL( \cH_{A} ) \to \cL( \cH_{A'} )$ and $\Gamma_{C} : \cL( \cH_{C} ) \to \cL( \cH_{C'} )$, where $\abs{ A' } = \abs{ C' } = 2$, such that
\extractabilitylb{,}
where $\beta_{b}$ is the violation of the $\CHSH_{b}$ inequality between Alice and Charlie conditioned on Bob observing the outcome $b$. Define
\begin{align*}
\sigABdef{&},\\
\sigBCdef{&}.
\end{align*}
Since $\Phi_{A'C'}^b$ are pure states and applying the channels $\Gamma_{A}, \Gamma_{C}$ commutes with the measurement performed on $B_{1} B_{2}$ we have
\begin{align*}
p_{b} F ( ( \Gamma_{A} &\otimes \Gamma_{C} ) ( \tau_{AC}^b ), \Phi_{A'C'}^b )\\
&= p_{b} \ave{( \Gamma_{A} \otimes \Gamma_{C} ) ( \tau_{AC}^b ), \Phi_{A'C'}^b }\\
&= \ave{ \sigma_{A' B_{1}} \otimes \sigma_{B_{2} C'}, \Phi_{A'C'}^b \otimes B_{ B_{1} B_{2} }^{b} },
\end{align*}
which, in particular, implies that
\begin{equation}
\label{eq:sigma-trace-lower-bound}
\ave{ \sigma_{A' B_{1}} \otimes \sigma_{B_{2} C'}, \Phi_{A'C'}^b \otimes B_{ B_{1} B_{2} }^{b} } \geq p_{b} g( \beta_{b} ).
\end{equation}

Recall that our goal is to construct a pair of unital maps $\Lambda_{ B_{1} } : \cL( \cH_{ B_{1} } ) \to \cL( \cH_{A'} )$ and $\Lambda_{ B_{2} } : \cL( \cH_{ B_{2} } ) \to \cL( \cH_{C'} )$ for which we can prove a lower bound on
\begin{equation*}
\ave{ ( \Lambda_{ B_{1} } \otimes \Lambda_{ B_{2} } ) ( B_{ B_{1} B_{2} }^{b} ) , \Phi_{A'C'}^b }.
\end{equation*}
If $\lambda_{A' B_{1}}$ and $\lambda_{B_{2} C'}$ denote the Choi states of the maps $\Lambda_{B_{1}}$ and $\Lambda_{B_{2}}$, respectively, we have
\begin{align*}
( \Lambda_{ B_{1} } \otimes &\Lambda_{ B_{2} } ) ( B_{ B_{1} B_{2} }^{b} )\\
&= \trj_{ B_{1} B_{2} } \big[ \big( \lambda_{A' B_{1}} \otimes \lambda_{B_{2} C'} \big) \big( \unit_{A'C'} \otimes ( B_{ B_{1} B_{2} }^{b} )\tran \big) \big]
\end{align*}
and therefore
\begin{align}
\ave{ ( \Lambda_{ B_{1} } \otimes &\Lambda_{ B_{2} } ) ( B_{ B_{1} B_{2} }^{b} ) , \Phi_{A'C'}^b }\nonumber\\
&= \ave[\big]{ \lambda_{A' B_{1}} \otimes \lambda_{B_{2} C'} , \Phi_{A'C'}^b \otimes ( B_{ B_{1} B_{2} }^{b} )\tran }\nonumber\\
\label{eq:lambda-inner-product}
&= \ave[\big]{ \lambda_{A' B_{1}}\tran \otimes \lambda_{B_{2} C'}\tran , \Phi_{A'C'}^b \otimes B_{ B_{1} B_{2} }^{b} },
\end{align}
where in the second step we have used the fact that the Bell states are invariant under transposition (in the standard basis). The similarity between this expression and Eq.~\eqref{eq:sigma-trace-lower-bound} suggests that the Choi states $\lambda_{ A' B_{1} }$ and $\lambda_{ B_{2} C' }$ should be constructed from $\sigma_{A' B_{1}}$ and $\sigma_{B_{2} C'}$, respectively. The only remaining difficulty is the fact that the marginals $\sigma_{A'}$ and $\sigma_{C'}$ are not necessarily proportional to $\unit$. Let us first show how to bound the non-uniformity of these marginals from the observed Bell violations. Let us parametrize the marginal of $\sigma_{A'}$ by $\eta_{A} \in [0, 1]$ such that
\begin{equation*}
\spec( \sigma_{A'} ) = \bigg\{ \frac{ 1 - \eta_{A} }{2}, \frac{ 1 + \eta_{A} }{2} \bigg\}.
\end{equation*}
Let $\sigma_{A'C'}^b := ( \Gamma_{A} \otimes \Gamma_{C} ) ( \tau_{AC}^b )$ and note that
\begin{equation*}
F( \sigma_{A'C'}^b, \Phi_{A'C'}^b ) = \ave{ \sigma_{A'C'}^b, \Phi_{A'C'}^b } \geq g( \beta_{b} ).
\end{equation*}
Let $U_{C'}^b$ be a local unitary acting on $\cH_{C'}$ such that $\big( \unit_{A'} \otimes U_{C'}^b \big) \ket{ \Phi^{0} }_{A'C'} = \ket{ \Phi^b }_{A'C'}$ and define
\begin{equation*}
\sigma_{A'C'}' := \sum_{b} p_{b} \big( \unit_{A'} \otimes U_{C'}^b \big)^{\dagger} \sigma_{A'C'}^b \big( \unit_{A'} \otimes U_{C'}^b \big),
\end{equation*}
where the summation goes over $b \in \{0, 1, 2, 3\}$. It is easy to verify that
\begin{align}
\label{eq:fidelity-lower-bound}
\begin{split}
F( \sigma_{A'C'}', \Phi_{A'C'}^{0} ) &= \ave{ \sigma_{A'C'}', \Phi_{A'C'}^{0} } = \sum_{b} p_{b} \ave{ \sigma_{A'C'}^b, \Phi_{A'C'}^b }\\
&\geq \sum_{b} p_{b} g( \beta_{b} ) = g( \betaave ) = q,
\end{split}
\end{align}
where we have used the fact that the function $g$ is linear. Moreover,
\begin{equation*}
\sigma_{A'}' = \sum_{b} p_{b} \sigma_{A'}^b = \sigma_{A'},
\end{equation*}
which implies that the two have the same spectrum. The lower bound given in Eq.~\eqref{eq:fidelity-lower-bound} plugged into Lemma~\ref{lem:fidelity-marginal-tradeoff} implies that $\eta_{A'} \leq \eta^{*}$ for
\begin{equation*}
\etastardef.
\end{equation*}
It is easy to check that for $\betaave > 2$, we have $\eta^{*} < 1$, i.e.~the reduced state $\sigma_{A'}$ is full-rank. By symmetry the same bound applies to $\eta_{C'}$.

We are now ready to define the Choi states of the channels $\Lambda_{ B_{1} }$ and $\Lambda_{ B_{2} }$. Let
\begin{align}
\lambda_{ A' B_{1} }\tran &:= \big( \sigma_{A'}^{-1/2} \otimes \unit \big) \sigma_{A'B_{1}} \big( \sigma_{A'}^{-1/2} \otimes \unit \big),\nonumber\\
\label{eq:lambdaB2C-def}
\lambda_{ B_{2} C' }\tran &:= \frac{2}{ 1 + \eta_{C'} } \sigma_{ B_{2} C' } + \sigma_{ B_{2} } \otimes \bigg( \unit - \frac{2}{ 1 + \eta_{C'} } \sigma_{C'} \bigg),
%
%
\end{align}
which are easily verified to be valid Choi states. For this particular choice we have
\begin{align*}
\lambda_{A' B_{1}}\tran &\otimes \lambda_{B_{2} C'}\tran\\
&\geq \lambda_{A' B_{1}}\tran \otimes \frac{2}{ 1 + \eta_{C'} } \sigma_{ B_{2} C' }\\
&\geq \bigg( s( \eta_{A'} ) \sigma_{ A'B_{1} } - t( \eta_{A'} ) \frac{ \unit }{2} \otimes \sigma_{ B_{1} } \bigg) \otimes \frac{2}{ 1 + \eta_{C'} } \sigma_{ B_{2} C' },
\end{align*}
where in the first step we use the fact that the second term in Eq.~\eqref{eq:lambdaB2C-def} is positive semidefinite, while in the second step we use the operator inequality derived in Lemma~\ref{lem:rescaling-marginal}. To use this operator inequality to bound the inner product given in Eq.~\eqref{eq:lambda-inner-product}, we note that the first term can be bounded using Eq.~\eqref{eq:sigma-trace-lower-bound}, whereas the second term can be explicitly evaluated
\begin{align*}
\ave{ &\unit_{A'} \otimes \sigma_{B_{1}} \otimes \sigma_{B_{2} C'}, \Phi_{A'C'}^b \otimes B_{ B_{1} B_{2} }^{b} }\\
&= \frac{1}{2} \ave{ \sigma_{B_{1}} \otimes \sigma_{B_{2}}, B_{ B_{1} B_{2} }^{b} } = \frac{1}{2} \ave{ \tau_{B_{1}} \otimes \tau_{B_{2}}, B_{ B_{1} B_{2} }^{b} } = \frac{ p_{b} }{2}.
\end{align*}
Combining these two results yields
\begin{align*}
\ave[\big]{ \lambda_{A' B_{1}}\tran \otimes \lambda_{B_{2} C'}\tran , &\Phi_{A'C'}^b \otimes B_{ B_{1} B_{2} }^{b} }\\
&\geq \frac{ p_{b} }{ 2 ( 1 + \eta_{C'} ) } \big[ 4 s( \eta_{A'} ) g( \beta_{b} ) - t(\eta_{A}) \big],
\end{align*}
which implies
\begin{align*}
Q( \sB,\Phi ) &\geq \frac{1}{4} \sum_{ b = 0 }^{3} \ave[\big]{ \lambda_{A' B_{1}}\tran \otimes \lambda_{B_{2} C'}\tran , \Phi_{A'C'}^b \otimes B_{ B_{1} B_{2} }^{b} }\\
&\geq \frac{ 4 s( \eta_{A'} ) \sum_{b} p_{b} g( \beta_{b} ) - t(\eta_{A'}) }{ 8 ( 1 + \eta_{C'} ) }\\
&= \frac{ 4 s( \eta_{A'} ) g( \betaave ) - t(\eta_{A'}) }{ 8 ( 1 + \eta_{C'} ) }\\
&= \frac{ 4 s( \eta_{A'} ) q - t(\eta_{A'}) }{ 8 ( 1 + \eta_{C'} ) }.
\end{align*}
This bound still depends on $\eta_{A'}$ and $\eta_{C'}$ and in order to remove this dependence we must minimize over $\eta_{A'}, \eta_{C'} \in [0, \eta^{*}]$. Since for $q \geq \frac{1}{2}$ the numerator is strictly positive, the minimization over $\eta_{C'}$ reduces to simply setting $\eta_{C'} = \eta^{*}$, which leads to the main result of the theorem.
\end{proof}
\section{Noise tolerance for the GHZ measurement}

In the previous section, we have shown that Theorem 1, in which we self-test a BSM in the case of ideal statistics, can be turned into Theorem 2, which is noise tolerant.
We explicitly presented a generalization of Theorem 1 to tilted BSM and GHZ measurement and discussed further generalization to other cases. In the following, we argue that Theorem 3 (self-testing of the GHZ measurement) can be made noise tolerant using the same approach. We expect this method generalizes to other cases.

Our method is based on existing results about the self-test of the GHZ state $\GHZ^{0,+}$ with Mermin inequality. For a tripartite state $\tau_{ABC}$ maximally violating Mermin inequality, it yields a product channel $\Gamma=\Gamma_A\otimes\Gamma_B\otimes\Gamma_C$ which maps $\tau_{ABC}$ to $\GHZ^{0,+}$. 
We first prove in Step 1 that after Rob obtained result $r$, the post measured state is mapped to the appropriate version of the $\GHZ$ state with $\Gamma$ i.e.~$\Gamma(\tau^r_{ABC})=\GHZ^r$. 
Without loss of generality, we assume that $\Gamma_A,\Gamma_B,\Gamma_C$ are locally applied by Alice, Bob, Charlie before Rob's measurement: now these three parties have a qubit each and share $\sigma_{A'R_A},\sigma_{B'R_B},\sigma_{C'R_C}$ with Rob such that $\sigma^r_{A'B'C'}=\GHZ^r$ after Rob measured $r$.

Then, in Step 2, we introduce the Choi-Jamio{\l}kowski isomorphism $\Lambda=\Lambda_{R_A}\otimes\Lambda_{R_B}\otimes\Lambda_{R_C}$ associated to the state $\sigma_{A'R_A}\otimes\sigma_{B'R_B}\otimes\sigma_{C'R_C}$ (up to normalization). By construction the measurement operators of Rob $R_{ R_{A} R_{B} R_{C} }^{r}$ are mapped to the post measured state shared between, Alice, Bob, Charlie:  $\Lambda( R_{ R_{A} R_{B} R_{C} }^{r} )=\sigma^r_{A' B' C'}$. As the marginal states $\sigma_{A'},\sigma_{B'},\sigma_{C'}$ are maximally mixed, the channel $\Lambda$ is unital, which proves Theorem 3.

In the noisy case, the proof has the same structure. As the self-test of $\GHZ^{0,+}$ with Mermin inequality is noise tolerant, for a tripartite state $\tau_{ABC}$ and sufficiently high violation of the Mermin inequality, there exist channels $\Gamma=\Gamma_A\otimes\Gamma_B\otimes\Gamma_C$ such that $\Gamma(\tau_{ABC})\approx\GHZ^{0,+}$. The rest of the proof directly applies. In particular, $\Gamma(\tau^r_{ABC})\approx\GHZ^r$ and if $\tilde\Lambda$ is the channel corresponding to $\sigma_{A' R_A}\otimes\sigma_{B' R_B}\otimes\sigma_{C' R_C}$ through the Choi-Jamio{\l}kowski isomorphism, we have $\tilde\Lambda( R_{ R_{A} R_{B} R_{C} }^{r} )=\sigma^r_{A'B'C'}$.  However, as the marginal states $\sigma_{A'},\sigma_{B'},\sigma_{C'}$ are no longer necessarily maximally mixed, $\tilde\Lambda$ may not be unital. Hence, we have to introduce new states $\lambda_{A'R_A},\lambda_{B'R_B},\lambda_{C'R_C}$ associated to a Choi-Jamio{\l}kowski isomorphism $\Lambda=\Lambda_{R_A}\otimes\Lambda_{R_B}\otimes\Lambda_{R_C}$ such that:
\begin{itemize}
\item $\lambda_{A'},\lambda_{B'},\lambda_{C'}$ are maximally mixed.
\item $\lambda_{A'R_A},\lambda_{B'R_B},\lambda_{C'R_C}$ are close to $\sigma_{A'R_A},\sigma_{B'R_B},\sigma_{C'R_C}$.
\end{itemize}
When this is the case, we have $\Lambda( R_{ R_{A} R_{B} R_{C} }^{r} )\approx \tilde\Lambda( R_{ R_{A} R_{B} R_{C} }^{r} )=\sigma^r_{A'B'C'}\approx\GHZ^r$ with $\Lambda$ unital, which proves the noisy variant of the theorem.

To find the new states $\lambda_{A'R_A},\lambda_{B'R_B},\lambda_{C'R_C}$, we can use the two constructions given in Eq.~(\ref{eq:lambdaB2C-def}). Their distance to $\sigma_{A'R_A}, \sigma_{B'R_B}, \sigma_{C'R_C}$ can be controlled in an analogous manner once we have established a bound on the bias of the marginals and this can be achieved by grouping two parties together. If we want to estimate the bias of Alice's marginal, we group Bob and Charlie together and we observe that now we have a maximally entangled two-qubit state between $A$ and $BC$, which allows us to use Lemma~\ref{lem:fidelity-marginal-tradeoff}. This proves a robust version of Theorem 3.

\section*{Appendix F: Choi-Jamio{\l}kowski isomorphism}\label{appendix_choi}

We recall here the definition of the Choi-Jamio{\l}kowski isomorphism.
Let $\rho_{AB}$ acting over $\cH_A\otimes\cH_B$ be a (possibly not normalized) bipartite state. Then its associated Choi-Jamio{\l}kowski isomorphism $\Gamma:\cH_B\rightarrow\cH_A$ is defined by the identity
\begin{equation}
\forall \sigma, \Gamma(\sigma)=\Trr{B}{\unit\otimes\sigma^T\cdot \rho_{AB}}.
\end{equation}
We have the following proposition, which can directly be deduced from the definition.
\begin{prop}\label{appendix_choirewritting}
Let $\rho_k$ acting over $\cH_{A_k}\otimes\cH_{B_k}$ and $\Gamma^k:\cH_{k}\rightarrow\cH_{k}$ be the associated Choi-Jamio{\l}kowski isomorphism. Let $\Omega$ be an operator of $\cH_{B}:=\bigotimes_k\cH_{k}$. Then
\begin{equation}
\bigotimes_k\Gamma^k (\Omega) = \Trr{B}{\Omega^T\cdot\bigotimes_k\rho_k}.
\end{equation}
\end{prop}
\begin{proof}
We introduce a decomposition $\Omega=\sum_{i} \bigotimes_k\omega_{k,i}$ where $\omega_{k,i}$ is an operator of $\cH_{k}$ and apply the definition of the Choi map:
\begin{align*}
\bigotimes_k\Gamma^k (\Omega)&=\bigotimes_k\Gamma^k (\sum_{i} \bigotimes_l\omega_{l,i})
=\sum_{i}\bigotimes_k\Gamma^k (\omega_{l,i})\\
&=\sum_{i}\bigotimes_k\Trr{\cH_{k}}{\unit_{A_k}\otimes\omega^T_{k,i}\cdot\rho_{k}}\\
&=\sum_{i}\Trr{\cH}{\bigotimes_k\unit_{A_k}\otimes\omega^T_{k,i}\cdot\rho_{k}}\\
&=\Trr{\cH}{\bigotimes_k\unit_{A_k}\otimes\Omega^T\cdot\rho_{k}}.
\end{align*}
\end{proof}
\ifthenelse{\equal{\shortversion}{1}}
{\end{comment}}{}

\end{document}